\documentclass{article}
\usepackage{graphicx} 

\usepackage[utf8]{inputenc}

\usepackage{amsmath,amsthm, amssymb}
\usepackage[margin=3cm]{geometry}
\usepackage{mathtools}
\usepackage{dsfont}
\usepackage{xcolor}
\usepackage{algorithm,algpseudocode}
\usepackage{todonotes}
\usepackage{nicefrac}
\usepackage{mathrsfs}
\usepackage{tikz}
\usepackage{authblk}
\usepackage{subcaption}

\usepackage{hyperref}

\tikzstyle{arrow} = [thick,->,>=stealth]

\usepackage{thm-restate}

\usepackage{etoc}

\usepackage{setspace}

\newtheorem{theorem}{Theorem}
\newtheorem{lemma}[theorem]{Lemma}
\newtheorem{corollary}[theorem]{Corollary}
\newtheorem{definition}[theorem]{Definition}
\newtheorem{prop}[theorem]{Proposition}

\usepackage{todonotes}

\newcommand{\RR}{\mathbb{R}}
\newcommand{\CC}{\mathbb{C}}
\newcommand{\FF}{\mathbb{F}}

\DeclareMathOperator{\Tr}{Tr}
\usepackage{amsthm}

\newcommand{\ket}[1]{|#1\rangle}
\newcommand{\bra}[1]{\langle #1|}

\newcommand{\ketbra}[2]{| #1\rangle\! \langle #2|}

\newcommand{\brackets}[1]{\left[ #1 \right]}

\newcommand{\set}[1]{\left\{ #1 \right\}}

\newcommand{\ts}{\textsuperscript}

\newcommand{\trace}[1]{\Tr \brackets{ #1 }}
\newcommand{\partrace}[2]{\Tr_{#1} \brackets{ #2 }}

\newcommand{\hilb}{\mathcal{H}}

\newcommand{\identity}{\mathds{1}}

\DeclareMathOperator{\im}{Im}
\DeclareMathOperator{\supp}{supp}

\allowdisplaybreaks

\title{Interpreting Multipartite Entanglement through Topological Summaries}
\author[1]{Raghav Banka}
\author[2]{Matthew Hagan}
\author[1,3,4]{Nathan Wiebe}

\affil[1]{
University of Toronto, Dept. of Computer Science, Toronto ON, Canada
}
\affil[2]{
University of Toronto, Dept. of Physics, Toronto ON, Canada
}
\affil[3]{
Pacific Northwest National Laboratory, Richland WA, USA
}
\affil[4]{
Canadian Institute for Advanced Research, Toronto ON, Canada
}

\date{}

\begin{document}

\maketitle



\begin{abstract}
     The study of multipartite entanglement is much less developed than the bipartite scenario. Recently, a string of results have proposed using tools from Topological Data Analysis (TDA) to attach topological quantities to multipartite states. However, these quantities are not directly connected to concrete information processing tasks making their interpretations vague. We take the first steps in connecting these abstract topological quantities to operational interpretations of entanglement in two scenarios. The first is we provide a bound on the Integrated Euler Characteristic defined by Hamilton and Leditzky via an average distillable entanglement, which we develop as a generalization of the Meyer-Wallach entanglement measure studied by A. J. Scott in 2004. This allows us to connect the distance of an error correcting code to the Integrated Euler Characteristic. The second is we provide a characterization of a class of graph states containing the GHZ state via the birth and death times of the connected components and 1-dimensional cycles of the entanglement complex. In other words, the entanglement distance behavior of the first Betti number $\beta_1(\varepsilon)$ allows us to determine if a state is locally equivalent to a GHZ state, potentially providing new verification schemes for highly entangled states.
\end{abstract}

\section{Introduction}
Entanglement is a hallmark feature of quantum mechanics that allows for correlations between separated systems that go beyond those allowable through classical mechanics \cite{Horodecki_2009}. Its usefulness in quantum information theory is evident through its applications in communication protocols, error-correcting codes, and cryptographic schemes \cite{BB84} \cite{Metger2021selftestingofsingle} \cite{Devetak_2006_ERC}. When entanglement is present with operations restrictions, such as only allowing arbitrary \emph{local} operations and classical communication (LOCC) between subsystems, then entanglement can become a resource. Despite this importance, our understanding of entanglement is mostly based on bipartite entanglement or quantum correlations between two halves of a single quantum system.  In bipartite systems, many questions about the entanglement structure of a state are completely characterized by the von Neumann entropy.

 Multipartite entanglement on the other hand is much less studied \cite{horodecki2024multipartiteentanglement}. For multipartite systems, there are much more nuanced questions that can be asked. For instance, state teleportation in multipartite systems becomes more complicated as there are no clear sources and recipients.  In bipartite scenarios, we can declare that systems are maximally entangled if either of the reduced density matrices is maximally mixed. Loosely speaking, this tells us that the information stored in the state is across a global correlation, not locally stored in either subsystem. For multipartite states, this becomes murkier. Consider the GHZ state $\ket{\rm GHZ} = \frac{1}{\sqrt{2}} \ket{00\ldots0} + \frac{1}{\sqrt{2}} \ket{1 1\ldots1}$. When looking at a single subsystem the resulting state is maximally mixed, indicating some kind of global entanglement. However, when looking at a slightly larger subsystem, such as two qubits, the resulting reduced density matrix is $\rho_{1,2} = \frac{1}{2} \ketbra{00}{00} + \frac{1}{2} \ketbra{11}{11}$ which is \emph{not} maximally mixed. Even the $W$ state, characterized for 3 qubits as $\ket{\rm W_3 } = \frac{1}{\sqrt{3}}(\ket{100} + \ket{010} + \ket{001})$ has slightly larger von Neumann entropy for any 2 qubits reduced density matrix but many would consider the GHZ state to be more ``globally'' entangled. This raises the question of when can we consider a multipartite pure state to be ``maximally'' entangled. Further, is there even one single meaningful quantity that we can reduce multipartite entanglement to?

A framework has recently been developed to shed light on this issue, which uses the seemingly unrelated area of research known as Topological Data Analysis (TDA) \cite{computational_topology_book} to characterize multipartite entanglement developed by Hamilton and Leditzky \cite{hamilton2023probing}. TDA was developed in order to provide information about the topological structure of a data set, which is independent of small local noise in the data set. These techniques can be used for example to estimate how many ``holes" are present in a dataset across varying distance scales. The canonical example is if we have samples of coordinates from a 3-dimensional surface it would be nice to know if the samples come from a sphere compared to a torus. TDA generalizes spheres vs tori to higher dimensions and provides a robust analysis of the number of holes across distance scales to get an idea of what features are persistent in the dataset as opposed to transient noise. Further, TDA provides a single numerical summary of the underlying manifold that data is sampled from to give a topologically invariant quantity to the data, known as the Euler's Characteristic. In graphs this corresponds to the number of edges minus the number of vertices and in three-dimensional shapes it is the number of edges minus the number of vertices and minus the number of faces. This gives a quick necessary heuristic to tell if two shapes or graphs are topologically equivalent but is not sufficient to declare they are indeed the same.  However, the operational meaning of this topological information remains unclear from this work. 

Studying multipartite entanglement using TDA, specifically persistence homology, was initiated in 2018 by di Pierro et al. \cite{diPierro_2018} and continued in \cite{mengoni2019persistenthomologyanalysismultiqubit, bart2023}. Early work focused on using bipartite entanglement measures to build higher dimensional data complexes, known as Vietoris-Rips complexes, and then classify small instances, up to 5 qubits, based on the resulting homology groups. The novelty of Hamilton and Leditzky \cite{hamilton2023probing} is that they show how higher dimensional entropy \emph{functionals}, such as correlation functions built on Tsallis $q$-deformed entropies, can lead naturally to higher dimensional data complexes through the C\v{e}ch complex construction \cite{computational_topology_book}. This approach offers a more refined construction than the Vietoris-Rips complex used prior, and they are able to show how this new topological structure can capture previously identified metrics. One of their main results is that the resulting Euler Characteristic is exactly equal to a previously used multipartite entanglement measure known as the $n$-tangle. This not only gives a more mathematically sound justification for the use of the $n$-tangle, which was later championed in \cite{horodecki2024multipartiteentanglement}, specifically the 3-tangle \cite{Kietaev_topology_2006} \cite{Hosur_quantum_channels_2016} but because the Euler Characteristic is a blunt summary there remains plenty of useful topological information in the TDA pipeline. 

We continue these efforts at providing mathematically rigorous interpretations of multipartite entanglement through TDA with the goal of providing operational interpretations of the topological information. For example, in bipartite entanglement theory, the entropy of entanglement between two subsystems tells us the minimum fidelity of a teleportation protocol that can be carried out between the two, or the number of bell pairs that would be needed to prepare a given bipartite state.  We further show that persistence diagrams yielded by this process can be used as a signature of a quantum state by showing that they can be used to uniquely identify GHZ states, up to local operations.  These results constitute therefore a substantial advance towards showing that topological properties of entanglement structure can be used to characterize multipartite entanglement.  

\section{Overview of Results}

In this paper, we build upon the work of Hamilton and Leditzky \cite{hamilton2023probing}. Our motivation for this project was to investigate further applications and interpretation of topological invariants of quantum states—namely, the Integrated Euler's Characteristic and Betti numbers. In doing so, we try to improve our understanding of the operational relevance of these topological invariants. In this article, we make no modification to the methodology of constructing topological elements. For a more detailed overview, readers may refer to the original work. We provide a brief description of the methodology in Section \ref{sec:prelims}. In the preliminaries, we also discuss the correlation measures used to construct the topological structures, along with a short introduction to relevant topological concepts—specifically the topological invariants we utilize. 

In Section \ref{sec:distillable_entanglement}, we introduce an entanglement measure called the Average Distillable Entanglement (ADE), which turned out to be a modified version of a known measure discussed in \cite{Scott_k-uniform}. To demonstrate its relevance and utility, we talked about its applications in characterizing the $k$-uniformity of quantum states and identifying quantum error-correcting codes. We then show how a topological invariant—the Integrated Euler's Characteristic—can be employed to estimate the ADE by providing lower and upper bounds. Furthermore, we discuss the tightness of these bounds for the case when the input quantum state is absolutely maximally entangled. For such a scenario, the ADE is saturated and the upper bound scales as $\Theta(\sqrt{n})$. 

In Section \ref{sec:graph_states}, we investigate the Betti numbers obtained from our Topological Data Analysis (TDA) in the context of graph states. We examine whether certain graph structures give rise to distinctive Betti number distributions, revealing unique topological signatures associated with certain classes of graph quantum states. We find that the GHZ state exhibits a unique topological footprint. To support this result, we introduce and prove connections between linear entropy and the underlying graph structure of quantum states. In section \ref{sec:k-uniform-barcodes}, we conclude by discussing the structure of topological barcodes for $k$-uniform and absolutely maximally entangled (AME) states.
\section{Preliminaries}\label{sec:prelims}

\subsection{Entropic Measures}
\label{subsec:entropic_measures}

We study multipartite quantum systems over an $n$ qubit  Hilbert space $\hilb = \hilb_{1} \otimes \hilb_{2} \otimes \ldots \otimes \hilb_{n} \cong \CC^{2^n}$, where for simplicity we will assume each individual quantum system is a qubit $\hilb_{i} = \mathbb{C}^{2}$. We will frequently consider subsets of vertices, for example $J = \set{1, 2, 3}$ of the total set of vertices $\set{1,2, \ldots ,n}$ which we denote $A \coloneqq \set{1, 2, \ldots, n}$. We will exclusively work with quantum states that are pure over the total Hilbert space $\hilb_A$ and typically denote the global state as $\rho_A = \ketbra{\psi}{\psi}$. Reduced states over a subset of qubits $J \subseteq A$ will be denoted as $\rho_J = \partrace{J^C}{\rho_A}$, where $J^C = A \setminus J$.

We denote the von Neumann entropy $S(\sigma) = - \trace{\sigma \log \sigma}$, all logs in this paper will be base 2. We make heavy use of the Tsallis $q$-deformed entropies which is defined for any subset $J \subseteq A$ and any real number $q \neq 0, 1$ as
\begin{equation}
\label{eq:tsallis_entropy}
    S_{q}(J)_{\rho} = \frac{1}{1 - q} (\trace{\rho_J^q} - 1).
\end{equation}
For $q=2$ the Tsallis entropy is equal to the purity of the state (defined as $1 - \trace{\rho^2}$) and in the limit as $q \to 1$ we have $\lim_{q \to 1} S_q(A)_{\rho} = S(\rho)$. The entropy of entanglement is defined for a bipartition of a Hilbert space, say $J$ and $J^C$, and is given by the von Neumann entropy of the reduced density matrix $S(\rho_J) = S(\rho_{J^C})$, where the equality holds as the global state over $A = J \cup J^C$ is pure. For pure global states, this quantity is further equal to the entanglement of formation and the distillable entanglement, which we will make use of and denote as $E_D(J)$. This discussion is crucial for proving our results presented in section \ref{sec:distillable_entanglement}. It is also important to note that both the entanglement of entropy $E_D(J)$ and the Tsallis $q$-entropy are invariant under Local Unitaries (LU), these are unitaries that respect the bipartition defined by $J$ (e.g. $U_J \otimes U_{J^C}$ is local with respect to $J$). We will make use of this property in Section \ref{sec:graph_states}.

In addition to entropic quantities, we will also make use of correlational quantities, such as the bipartite mutual information, also known as the total correlation, which is defined for two disjoint subsets $J_1, J_2 \subset A$, letting $J = J_1 \cup J_2$, as
\begin{equation}
\label{eq:total_correlation_bipartite}
C(J_1, J_2) = S(J_1) + S(J_2) - S(J) = D(\rho_{J} || \rho_{J_1} \otimes \rho_{J_2}),
\end{equation}
where $D(\rho_{J_1 \cup J_2} || \rho_{J_1} \otimes \rho_{J_2})$ is the quantum relative entropy. This definition in terms of the quantum relative entropy makes it extensible to multipartite systems. Let $J_1, J_2, \ldots J_k \subset A$ be a collection of disjoint sets and let $J = \bigcup_{i = 1}^k J_k$. The total correlation of this collection is given by

 \begin{equation}
 \label{eq:multipartite_total_correlation}
     C(J_1, J_2, \ldots, J_k) \coloneqq D(\rho_{J} || \rho_{J_1} \otimes \rho_{J_2} \ldots  \otimes  \rho_{J_k}) = S(J_1) + S(J_2) + \dots + S(J_n) - S(J).
 \end{equation}

It is important to note that relative entropy is not a true distance metric, as it is neither symmetric nor does it satisfy the triangle inequality \cite{vedral2002relativeentropy}. To be able to form valid simplicial complexes, as described in the following section, we will not need either of these properties but instead, we will need a downward closure property. A functional $C$ (a map from the powerset $2^A$ to $\RR$ or $\CC$) satisfies downward closure if for sets $J, K \subseteq A$ we have
\begin{equation}
    J \subseteq K \implies C(J) \le C(K). \label{eq:downward_closure_total_correlation}
\end{equation}
This property can be readily seen to hold for the total correlation $C$ due to the monotonicity of relative entropy under any arbitrary channel $\Phi$, namely $D(\Phi(\rho) || \Phi(\sigma)) \le D(\rho || \sigma)$. Downward closure then follows from $\Phi = {\rm Tr}_{K \setminus J} [ ~\cdot~ ]$. 

In the last section of this paper, Section \ref{sec:graph_states}, we will use an alternative total correlation functional, namely the 2-deformed total correlation which is given by
\begin{equation}
 \label{eq:multipartite_2-deformed_total_correlation}
C_2(J_1, J_2, \dots, J_k) \coloneqq S_2(J_1) + S_2(J_2) + \ldots + S_2(J_k) - S_2(J),
\end{equation}
where all $J_i$ are disjoint and $J = \bigcup_{i=1}^k J_i$. This function is typically not defined in terms of a $q$-deformed relative entropy but instead via the difference of single system Tsallis entropies and the total overall Tsallis entropy. Proof of the downward closure property for this correlation functional is given in \cite{hamilton2023probing} and is obtained using subadditivity properties of entropy functions.

\subsection{Topological Data Analysis}
\label{subsec:topology}

Topological Data Analysis (TDA) is a set of techniques for inferring topological information of an underlying manifold from data that is assumed to be noisy. The most common tool in the TDA toolkit is persistence homology, which tells us which ``holes" present in the dataset are longest lived, which could be indicative of significant noise in the underlying dataset. Recently this tool has been extended to quantum information theory through a string of works \cite{diPierro_2018, mengoni2019persistenthomologyanalysismultiqubit, bart2023, hamilton2023probing}. In this section, we provide a brief review of how to construct the relevant topological objects and quantities on which our results are built.

We will store the entanglement information present in our $n$-qubit state in a simplicial complex $X$, sometimes called the entanglement complex. A simplicial complex starts with a collection of vertices, denoted $X^{(0)}$, where the 0 stands for 0-dimensional, and we have a single vertex for each qubit, giving $X^{(0)} = \set{1, 2, \ldots, n}$. We then construct faces, or subsets of vertices, that are organized by their cardinality. A one-dimensional face consists of two vertices $\set{v_1, v_2}$ and in general a $d$-dimensional face is a subset of $X^{(0)}$ of cardinality $d + 1$. We denote the collection of all $d$ dimensional faces in the complex as $X^{(d)}$. The requirement that the object form a simplicial complex boils down to requiring downward closure, if a $d$-dimensional face $\sigma \in X^{(d)}$ then all subsets of $\sigma$ must also be in $X$, and the intersection property that the intersection of two faces must also be a face.

Once a simplicial complex $X$ is built we can define a $ k $-chain as a linear combination of $k$-simplices, where the linear combination is performed over the field $\FF_2$. Other fields such as $\FF_p$ or $\RR$ can be used, but we will use $\FF_2$ in line with prior work. We essentially treat all faces in $X^{(k)}$ as a basis element for the vector space of $k$-chains which allows us to write a $k$-chain as
$$
\sigma_1 + \sigma_2 + \ldots + \sigma_m \in X^{(k)},
$$
where each $\sigma_i$ is a face in $X^{(k)}$. We will denote the vector associated with a set with square brackets, e.g. the vector associated with the face $\set{1, 2,3}$ is $[1, 2,3]$. We can then treat each dimension of faces $X^{(k)}$ as a separate vector space. The distinction of $X^{(k)}$ as a vector space or as an abstract collection of faces is unimportant in this paper and we will be explicit about our use of $X^{(k)}$ if need be, but we will make more use of the vector space denotation.

Once we have these vector spaces $X^{(k)}$ we now define the boundary operators $\partial_k : X^{(k)} \to X^{(k-1)}$ between them. We will describe the action of the $\partial_2$ intuitively before giving an abstract definition. Given a triangle $[1,2,3]$ the boundary operator $\partial_2$ maps it to the ``boundary", which is a linear combination of the lines between the vertices 
$$\partial_2 [1,2,3] = [1,2] + [2, 3] + [1, 3].$$
To then generalize this, given a $d$-dimensional face $[x_0, x_1, \ldots, x_d]$ we denote the removal of an element from this face with a $~\hat{}~$ operator, see \cite{Hatcher:478079}. So $[1, \hat{2}, 3] = [1, 3]$. The $k$\ts{th} boundary operator is then the linear combination of removing a single element from the face over all possible elements:
\begin{equation}
\partial_k [x_0, x_1, \cdots, x_k] = \sum_{i = 0}^k  [x_0, \cdots, \hat{x_i}, \cdots, x_k].
\end{equation}
This action is then extended by linearity to non-basis vectors. One special type of $k$-chain to note is those that have a vanishing boundary, $\partial_k \sigma = 0$. If a $k$-chain $\sigma$ is in the kernel of $\partial_k$ it is called a $k$-cycle.

This collection of vector spaces and maps is a chain complex, written as
\begin{equation}
    X^{(n)} \xrightarrow{\partial_n} X^{(n-1)} \xrightarrow{\partial_{n-1}} \ldots X^{(1)} \xrightarrow{\partial_1} X^{(0)} \xrightarrow{\partial_0} 0,
\end{equation}
if we also enforce the homology condition: $\partial_{k + 1} \circ \partial_{k} = 0$. This can be rewritten as requiring $\im \partial_{k + 1} \subset \ker \partial_{k}$. As $\im \partial_{k + 1}$ and $\ker \partial_{k}$ are subspaces of the vector space $X^{(k)}$, we can consider the quotient space 
\begin{equation}
 H_{k} \coloneqq  \frac{\ker \partial_{k}}{\im \partial_{k + 1}}   
\end{equation}
 which is called the $k$\ts{th} homology group. Its dimension as a vector space is an important number, the $k$\ts{th} Betti number $\beta_k \coloneqq \dim(H_k) = \dim(\ker \partial_k) - \dim(\im \partial_{k + 1})$. The alternating sum of these Betti numbers gives the Euler Characteristic 
\begin{equation}
    \chi = \sum_{i = 0}^{k} (-1)^i \beta_i.
\end{equation}




Now that we have defined the main quantities of interest we can describe the TDA pipeline, which is sketched in Fig. \ref{fig:tda_pipeline}.
The analysis begins with a quantum state as input and a filtration parameter $\varepsilon$ for the total correlation $C$, where we remind the reader that we use either the von Neumann entropy-based total correlation or the linear entropy-based total correlation. We then build the complex from the lowest dimension to the highest, starting with 1-dimensional faces. For these 1-dimensional faces we iterate over all possible pairs $J = \set{j_1, j_2}, \forall j_1, j_2 \in A$, and add the face to the complex if $C(J) \le \varepsilon$. Once the 1-dimensional faces are determined we can repeat this process for higher dimensions by simply adding more elements to the sets $J$ until the largest dimension of interest is reached. We note that the downward closure property of $C$, along with the requirement that $J \in X^{(k)} \implies C(J) \le \varepsilon$, is sufficient to guarantee that the resulting structure is a simplicial complex. 

Once we have a simplicial complex $X(\rho, \varepsilon)$ we can then compute topological information about it, such as the Betti numbers $\beta_i$ and the Euler Characteristic $\chi$, which give a quantitative estimate of the ``connectedness" of the complex. The $0$-dimensional holes in the topological structure represent the number of connected components. If $X$ has a large $\beta_0$ this tells us there are a lot of disconnected pieces of the complex. A large $\beta_1$ tells us there are a lot of one-dimensional holes(corresponding to cycles) in the complex, similar to glueing a bunch of donuts together to tile a big square.  Although these numbers for a specific value of $\varepsilon$ may be useful, the real utility of TDA comes when considering how these topological quantities change as a function of $\varepsilon$. If a $k$-cycle is present only for a small range of $\varepsilon$ then it is most likely a byproduct of noise. However, long-lived cycles indicate features of the underlying state $\rho$ that are persistent across a large range of correlation values and are most likely important to the entanglement structure of $\rho$. 

As there is a maximum value for total correlation $C$ at a fixed maximum dimension, we denote the value of $\varepsilon$ where all possible faces are present in the complex as $\varepsilon_{\max}$. We refer to the values of $\varepsilon$ that a $k$-cycle appears and disappears as its birth time and death time, respectively. The study of how the topological features of $X(\rho, \varepsilon)$ vary across the entire range $[0, \varepsilon_{\max} ]$ is known as persistence homology \cite{computational_topology_book}. One of the key outputs is a plot of the birth and death times of each cycle as a function of $\varepsilon$ and is referred to as a barcode. We will make use of a more persistent variant of the Euler Characteristic called the Integrated Euler Characteristic (IEC) denoted by $\tilde{\mathfrak{X}}$. This was introduced in \cite{hamilton2023probing} and requires a variant of the 0-dimensional Betti number. Define $\tilde{\beta}_0 = \max \set{\beta_0 - 1, 0}$, which is essentially the normal 0-dimensional Betti number but subtracted off a single connected component. This eliminates divergences due to the presence of one connected component at large values of $\varepsilon$. Then the IEC is defined for a $d$-dimensional entanglement complex $X(\rho, \varepsilon)$ as
\begin{equation}
    \tilde{\mathfrak{X}}(\varepsilon) \coloneqq \int_0^\varepsilon \widetilde{\beta}_0 (\varepsilon' ) d\varepsilon' + \sum_{k = 1}^d (-1)^k \int_0^{\varepsilon} \beta_k (\varepsilon') d\varepsilon'.
\end{equation}
We will tie this construction to other existing entanglement notions in Section \ref{sec:distillable_entanglement}. One of the main results of \cite{hamilton2023probing}, namely Theorem 4.18, is that the IEC is equal to the following
\begin{align}
    \tilde{\mathfrak{X}}(\varepsilon) = \sum_{J \subseteq \mathcal{A}} (-1)^{|J| - 1} S_q (J)_\rho.
\end{align}
This is the expression that we will make the most use of.

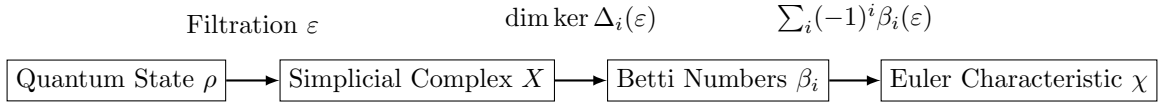
\begin{figure}
    \centering
    \begin{tikzpicture}[main/.style = {draw, rectangle}] 
        \node[main] (1) at (0, 0) {Quantum State $ \rho $};
        \node[main] (2) at (4, 0) {Simplicial Complex $ X $};
        \node[main] (3) at (8, 0) {Betti Numbers $ \beta_i $};
        \node[main] (4) at (12, 0) {Euler Characteristic $ \chi $};
        \draw[->, thick, >=latex] (1) -- node[above, yshift=0.5cm]{Filtration $ \varepsilon $} (2);
        \draw[->, thick, >=latex] (2) -- node[above, yshift=0.5cm]{$ \dim \ker \Delta_i(\varepsilon) $} (3);
        \draw[->, thick, >=latex] (3) -- node[above, yshift=0.5cm]{$ \sum_i (-1)^i \beta_i(\varepsilon) $} (4);
    \end{tikzpicture}
    \caption{Basic TDA pipeline}
    \label{fig:tda_pipeline}
\end{figure}
\section{Entanglement Distillation} 
\label{sec:distillable_entanglement}

For a generic mixed state $\rho$ over a bipartite system $\hilb_{J} \otimes \hilb_{J^C}$ there are many different notions of entanglement we can define. The first is the entropy of entanglement, given as $S(\rho_{J}) = - \trace{\rho_{J} \log \rho_{J^C}}$. This is rigorously defined but does not have an operational interpretation on its own. One of the more useful definitions we will use is the question ``How many Bell pairs can be extracted from the state $\rho$?" Bell pairs are useful in many contexts, one of their primary uses is in teleporting single qubit states. State teleportation is a special variant of gate teleportation, which allows us to create a resource state, such as a ``magic state" which encodes a $T$-gate, on one half of a Bell pair and then applies the gate to the other half. Heuristically this process looks like $U \ket{0} \ket{T} \mapsto T U \ket{0} \ket{0}$. This is a critical technique for quantum error correction to avoid the Easton-Knill theorem, which states that the set of logical operations available in an error-correcting code cannot be universal. $T$-gates are the most common gates that cannot be applied in many existing codes, making the production of magic $T$-states via distillation and consuming them via gate teleportation a vital routine in fault-tolerant quantum computing.

Given a quantum state, two of the most common operational questions asked are how many Bell pairs are needed to form the state and how many Bell pairs can be extracted from the state. We focus on the latter in this section. Formally this process can be defined as follows.

\begin{definition} [Distillable Entanglement $E_D(\rho_{AB})$]
Let $\rho = \ketbra{\psi}{\psi}$ be a pure quantum state over a Hilbert space $\mathcal{H}$ with a bipartition $\mathcal{H} = \mathcal{H}_A \otimes \mathcal{H}_B$. Let $\Phi_{A \times B}^{(n)} \in LOCC((\mathcal{H}_A \otimes \mathcal{H}_B)^{\otimes n})$ be a family of ``distilling" channels that utilize only LOCC operations between asymptotically many copies of $\mathcal{H}_A$ and $\mathcal{H}_B$ that aim to create copies of Bell pairs $\tau = \left(\frac{\ket{0}\ket{0} + \ket{1}\ket{1}}{\sqrt{2}} \right) \left(\frac{\bra{0}\bra{0} + \bra{1}\bra{1}}{\sqrt{2}} \right)$. The \emph{distillable entanglement} is then the supremum value of all integers $m$ such that. 
$$
\lim_{n \to \infty} F\left( \tau^{\otimes m}, \Phi_{A \times B}^{(n)} (\rho^{\otimes n}) \right) = 1.
$$

\end{definition}
Based on this idea of entanglement distillation, we propose an entanglement measure called the average distillable entanglement.
\begin{definition} [Average Distillable Entanglement]
\label{def:Average_Distillable_Entanglement}
    Let $\rho = \ketbra{\psi}{\psi}$ be a pure multipartite quantum state with qubit set $\mathcal{A} := \{\mathcal{A}_1, \ldots \mathcal{A}_n\}$. The average distillable entanglement is the mean value of the distillable entanglement calculated over all possible bipartitions of $\mathcal{H_A} = \mathcal{H}_J \otimes \mathcal{H}_{J^c}$ where $J \subseteq \mathcal{A}$,
    \begin{equation}
     \langle E_D \rangle = \frac{\sum_{\substack{J \subseteq \mathcal{A}} }E_D(\rho_{J})}{2^n}.
    \end{equation}
    
\end{definition}

We can interpret the average distillable entanglement as the number of Bell pairs that can be generated through LOCC operations averaged over all possible partitions of a multipartite quantum state. A similar entanglement measure has been studied in previous works, such as in \cite{Scott_k-uniform} and \cite{meyer2001global}. Specifically, in \cite{Scott_k-uniform}, the authors used linear entropy instead of von Neumann entropy, but both measures share the same core principle: averaging of a known bipartite entropy measure across all partitions of a multipartite quantum system. The average distillable entanglement quantity can be considered an entanglement measure as satisfies most of the key properties of entanglement, namely invariance under local operations and monotonicity under LOCC (Local Operations and Classical Communication). However, it may take values outside the range $[0, 1]$. We can always normalize this measure to ensure it satisfies the constraint that it lies within this range. This quantity is useful in determining whether a quantum state is an AME state and also provides a necessary condition for $k-$uniform states, which we will define and discuss formally in the next section.

\subsection{k-uniform states}
We will further build upon the applications of this measure to identify specific entanglement features of the quantum state, specifically the $k$-uniformity of the state. This special class of quantum states has been studied in the past under various names such as $k$-MM states \cite{Cerf_k-MM}, $k$-uniform states \cite{Karol_k-uniform_construction}, and normal forms \cite{MoorNormal_form}. They are particularly useful in quantum error correction codes \cite{Scott_k-uniform} and studying the effects of local unitary operations in quantum systems \cite{Kraus_LU_Equivalence}. A special case of a $k$-uniform system is the $\lfloor \frac{n}{2} \rfloor$-uniform state, known as the Absolutely Maximally Mixed State (AME). Because of its highly entangled structure, this class of quantum state has proven to be important in ADS/CFT protocols, holographic codes \cite{preskill_holographic_codes}, secret sharing \cite{Helwig_AME}, and teleportation protocols \cite{Helwig:2013ckb}. We will now formally define $k$-uniform quantum systems and discuss how they relate to our entanglement measure of average distillable entanglement. 

\begin{definition}
\label{def:k-uniform}
($k$-uniform state)
    Let $\rho = \ketbra{\psi}{\psi}$ be a pure multipartite quantum state defined over set of qubits $\mathcal{A} := \{\mathcal{A}_1, \ldots \mathcal{A}_n\}$. We say that the quantum state $\rho$ is $k-$uniform if 
    \begin{equation}
        \rho_J = \rho_{J^c} = \Tr_{J^c}(\ketbra{\psi}{\psi}) = \frac{\identity}{2^k}
    \end{equation}
for each non empty subset $ J \subset \mathcal{A} $ with $ |J| = k $, and where $ J^c \coloneqq \mathcal{A} \setminus J $, for $ k $ in $ 1 $ to $ \left\lfloor \frac{n}{2} \right\rfloor $.
\end{definition} 
This definition naturally leads to the definition of an AME state, which we give below.
\begin{definition}
    (Absolutely maximally entangled states)
    \label{def:AME_state}
    These are $k - $uniform quantum states where $k = \lfloor \frac{n}{2} \rfloor$. We say a pure quantum state is Absolutely Maximally Entangled (AME) if it is maximally mixed across all bipartitions. More mathematically,   Let $\rho = \ketbra{\psi}{\psi}$ be a pure multipartite quantum state defined over set of qubits $\mathcal{A} := \{\mathcal{A}_1, \ldots \mathcal{A}_n\}$. We say that the quantum state $\rho$ is an absolutely maximally entangled state if 
        \begin{equation}
        \rho_J = \rho_{J^c} = \Tr_{J^c}(\ketbra{\psi}{\psi}) = \frac{\identity}{2^{min({|J|, n - |J|})}}
    \end{equation}
    for every non-empty subset $ J \subset \mathcal{A}, J^c = \mathcal{A} \setminus J$.
\end{definition}
It is helpful to first explain how we compute the von Neumann entropy of AME states, as this will be used in the proofs of subsequent propositions and theorems. For any subset of qubits $ J $ such that $ |J| = k \leq \left\lfloor \frac{n}{2} \right\rfloor $, the von Neumann entropy of the subset $ J $ is given by $ |J| \log d = k \log 2 = k $. This property of AME states is also discussed in \cite{helwig2013quditgraph}. On the other hand, the von Neumann entropy of any subset of qubits with size $ |J| = k > \lfloor \frac{n}{2} \rfloor $ is given by $ \log 2^{n - k} = n - k $.

\begin{prop}
$\ket{\psi}$ is a product state if and only if the average distillable entanglement of $\rho = \ketbra{\psi}{\psi}$ is zero.
\end{prop}
\begin{proof}
The product states have an entropy of zero for all their subsystems, since partial tracing over any subsystem results in a pure state. Consequently, the average distillable entanglement is also zero. On the other hand, if the average distillable entanglement is zero, it implies that there is no entanglement across any partition of the quantum state and the entropic value of all subsystems of the quantum state is 0. This indicates that the quantum state must be a product state.
\end{proof}
\begin{prop}
\label{prop:ADE_iff_AME}
A pure quantum state $\ket{\psi}$ is an AME state defined for qubit set $\mathcal{A} := \{\mathcal{A}_1, \ldots \mathcal{A}_n\}$ if and only if the average distillable entnaglement of $\rho = \ketbra{\psi}{\psi}$ is given by
\begin{equation}
    \langle E_D \rangle = \frac{n}{2^n} \left(2^{n-1} - \binom{n-1}{\lfloor \frac{n}{2} \rfloor} \right).
\end{equation}





\end{prop}
\begin{proof}
    Assume $\ket{\psi}$ is $\lfloor \frac{n}{2} \rfloor$-uniform state. This implies that $\ket{\psi}$ is $k-$uniform for all $k < \lfloor \frac{n}{2} \rfloor$. Thus, we have that $S(J) = k$ for all $J \subseteq A$ where $|J| = k < \lfloor \frac{n}{2} \rfloor$. Since the state is AME (absolutely maximally entangled), the reduced density matrix becomes proportional to the identity matrix, leading to the maximum entropy for all local degrees of freedom. The average distillable rate of the state follows the equation for states with an odd number of qubits.
    
\begin{equation}
\langle E_D \rangle = 2^{-n} \sum_{J \subseteq \mathcal{A}} S(J) = 2^{-n} \sum_{i=1}^{n-1} \sum_{\substack{J \subset \mathcal{A} \\ |J| = i}} S(J)
\end{equation}

Note that we only need to concern ourselves with proper subsets, as the entropy of the entire quantum state is 0, and we assume the state to be pure. 
By Schmidt decomposition across the partitions $J$ and $J^c$, we know the reduced density matrix (RDM) of $\rho$ over subsystem $J$ shares eigenvalues with the RDM of $\rho$ over subsystem $J^c$. Consequently, we can conclude that $S(J) = S(J^c)$ for all $J \subset \mathcal{A}$. Since the number of subsets $ J \subset \mathcal{A} $ with $ |J| = k $ is equal to the number of subsets with $ |J| = n - k $, it follows that for any $ k < \left\lfloor \tfrac{n}{2} \right\rfloor $,
\begin{equation}
\sum_{\substack{J \subset \mathcal{A} \\ |J| = k}} S(J) = \sum_{\substack{J \subset \mathcal{A} \\ |J| = n - k}} S(J) = \binom{n}{k}  k    
\end{equation}
Using this property we can simplify the average distillable entanglement expression.
For quantum states with an odd number of qubits, the average distillable entanglement is written as 
\begin{equation}
\langle E_D \rangle =  2^{-n} \sum_{i=1}^{n-1} \sum_{\substack{J \subset \mathcal{A} \\ |J| = i}} S(J) =  2^{-(n-1)} \sum_{i=1}^{\lfloor \frac{n}{2} \rfloor} \sum_{\substack{J \subset \mathcal{A} \\ |J| = i}} S(J) = \frac{1}{2^{n-1}} \sum_{i=1}^{\lfloor \frac{n}{2} \rfloor} \binom{n}{i}  i
\end{equation}
 We can further simplify this quantity by using the identity $\binom{n}{i}i = n \binom{n-1}{i-1}$. Let's assume $n = 2k +1$ for some integer $k$.
\begin{flalign}
    \langle E_D \rangle &= \frac{1}{2^{n-1}} \sum_{i=1}^{\lfloor \frac{n}{2} \rfloor} \binom{n}{i}  i = \frac{2k +1}{2^{2k}} \sum_{i=1}^{k} \binom{2k}{i-1} \\
    &=  \frac{2k +1}{2^{2k}} \sum_{j=0}^{k-1} \binom{2k}{j} = \frac{2k +1}{2^{2k}} \left[ \frac{1}{2} \left( 2^{2k} - \binom{2k}{k} \right) \right] \\
    &= \frac{n}{2^n} \left(2^{n-1} - \binom{n-1}{\lfloor \frac{n}{2} \rfloor} \right)
\end{flalign}
For an even number of qubits, we have $n = 2k$ for some integer $k$. Using a similar technique as above we get that
\begin{flalign}
    \langle E_D \rangle &= \frac{1}{2^n} \left( 2\sum_{i=1}^{\frac{n}{2}-1} \sum_{\substack{J \subset A \\ |J| = i}} S(J) + \sum_{\substack{J \subset A \\ |J| = \frac{n}{2}}} S(J) \right) = \frac{1}{2^n} \left[ 2 \sum_{i=1}^{\frac{n}{2} - 1} \binom{n}{i}  i + \binom{n}{\frac{n}{2}}   \frac{n}{2}  \right] \\
    &= \frac{1}{2^{2k}} \left[ 4k \sum^{k - 2}_{j=0} \binom{2k - 1}{j} + \frac{2k}{2}\binom{2k}{k}\right] \\
    &= \frac{1}{2^{2k}} \left[ \frac{4k}{2} \left( 2^{2k -1} - 2 \binom{2k-1}{k-1}\right) + \frac{2k}{2}\binom{2k}{k}\right] \\
    &= \frac{2k}{2^{2k}} \left[ 2^{2k-1} - 2\binom{2k-1}{k} + \frac{1}{2} \binom{2k}{k}\right] \\
    &= \frac{2k}{2^{2k}} \left[ 2^{2k-1} - 2\binom{2k-1}{k} + \frac{1}{2} \left( \binom{2k - 1}{k} + \binom{2k - 1}{k - 1}\right)\right] \\
        &= \frac{2k}{2^{2k}} \left[ 2^{2k-1} - 2\binom{2k-1}{k} + \frac{1}{2} \left( \binom{2k - 1}{k} + \binom{2k - 1}{k}\right)\right] \\
    &= \frac{2k}{2^{2k}} \left[ 2^{2k-1} - \binom{2k-1}{k}\right] \\
    &= \frac{n}{2^n} \left(2^{n-1} - \binom{n-1}{\lfloor \frac{n}{2} \rfloor} \right)
\end{flalign}

We know that the combinatorial expression proposed in Proposition~\ref{prop:ADE_iff_AME} is the maximum value of the average distillable entanglement because each term in the sum used to compute the average distillable entanglement is maximized if the input quantum state is AME. This follows from the fact that the entropy of any subsystem is maximal if and only if the reduced state is maximally mixed. To prove the reverse direction of our proposition, we will prove the contrapositive of the statement: if the state is not AME, then the average distillable entanglement is not maximized. Assume that the quantum state is not AME. By definition, there exists a subset $ J \subset \mathcal{A} $, with $ |J| = k \leq \left\lfloor \frac{n}{2} \right\rfloor $, such that the reduced density matrix $ \rho_J $ is not maximally mixed. It suffices to consider subsets $ J $ with $ |J| \leq \left\lfloor \frac{n}{2} \right\rfloor $, since if $ \rho_J $ is not maximally mixed for some $ |J| > \left\lfloor \frac{n}{2} \right\rfloor $, then its complement $ \rho_{J^c} $ (with $ |J^c| < \left\lfloor \frac{n}{2} \right\rfloor $) is also not maximally mixed, by the Schmidt decomposition. If $ \rho_J $ is not maximally mixed, then the von Neumann entropy of the reduced density matrix satisfies $ S(\rho_J) < |J| $, since the von Neumann entropy attains its maximum value if and only if the state is maximally mixed. As the average distillable entanglement is defined as the normalized sum of such entropies \ref{def:Average_Distillable_Entanglement}, we can conclude that the average distillable entanglement is strictly less than its maximum possible value.

\end{proof}
\begin{corollary}
\label{corollary:lower_bound_of_k-uniform}
    Let $\rho = \ket{\psi}\bra{\psi}$ be a pure $n$-qubit $k$-uniform state for some $k < \lfloor \frac{n}{2} \rfloor$(see Ref. \ref{def:k-uniform}). The average distillable entanglement must be atleast $\frac{1}{2^{n-1}} \sum_{i=1}^{k} \binom{n}{i} i$.
\end{corollary}

\begin{proof}

We assume that the input quantum state is $k-$uniform for $k < \lfloor \frac{n}{2} \rfloor$. With this assumption, we can lower bound the average distillable entanglement. We know that RDM $\rho_J$ and $\rho_{J^c}$is maximally mixed for subset of qubits $J$ where $|J| \leq k$. The Von Neumann entropy for such RDMs will be maximum as well. Using this fact in our average distillable entanglement equation, we get that 
\begin{equation}
     2^{-n} \left( \sum_{i=1}^{k} \sum_{\substack{J \subset \mathcal{A} \\ |J| = i}} S(J) + \sum_{i=n-k}^{n-1} \sum_{\substack{J \subset \mathcal{A} \\ |J| = i}} S(J)\right) \leq \langle E_D \rangle =  2^{-n} \sum_{i=1}^{n-1} \sum_{\substack{J \subset \mathcal{A} \\ |J| = i}} S(J) 
\end{equation}
\begin{align}
         2^{-n} \left( \sum_{i=1}^{k} \binom{n}{i}i + \sum_{i=n-k}^{n-1} \binom{n}{i}(n-i) \right) &\leq \langle E_D \rangle  \\
        2^{-n} \left( \sum_{i=1}^{k} \binom{n}{i}i + \sum_{j=1}^{k} \binom{n}{n - j}j \right) &\leq \langle E_D \rangle  \\
        2^{-(n -1)} \sum_{i=1}^{k} \binom{n}{i}i  &\leq \langle E_D \rangle 
\end{align}

In this case, there was no need to consider whether $ n $ is even or odd, because we were working with the $ k $-uniform case where $ k $ is strictly less than $ \left\lfloor \frac{n}{2} \right\rfloor $. For any subset of qubits $ J $, we can always find its complement $ J^c $ such that $ |J| \neq |J^c| $. We have shown that if $\rho$ is $k-$uniform then the ADE(Average Distillable Entanglement) must be atleast $\frac{1}{2^{n-1}} \sum_{i=1}^{k} \binom{n}{i} i$. Taking the contrapositive of the implication, we can also conclude that if the average distillable entanglement of the state is strictly lesser than $\frac{1}{2^{n-1}} \sum_{i=1}^{k} \binom{n}{i} i$, then the state is not $k-$uniform.
\end{proof}

\subsection{Error Correcting code}
We continue discussing the applications of average distillable entanglement. We turn our attention towards error-correcting codes, which play an integral role in making quantum computing feasible. These error-correcting codes protect quantum states from external noise and decoherence, which affect the reliability of the computation and information processing. By encoding the input quantum state, error-correcting codes we ensure that the quantum state can be reliably decoded, even after being subjected to noise. Specifically, we map an input state containing $k$ qubits, also known as physical qubits, into a higher-dimensional code space of $n$ qubits, referred to as logical qubits. A quantum code is said to correct up to $t$ errors if, after $t$ physical qubits are altered by noise, the original quantum state of $k$ qubits can still be faithfully recovered. The distance of a code space is defined as the minimum Hamming distance between any two codewords in the code space. An error-correcting code with a distance of $2t + 1$ can correct up to $t$ qubits affected by noise.

We will first define the error-correcting code more formally before demonstrating the connection between ADE and error-correcting codes.  Let $\mathcal{Q}$ represent an error-correcting code space spanned by the logical basis $\{\ket{i_L} \mid i = 0, 1, \ldots, K-1\}$. The code space is a $K$-dimensional subspace of an $n$-dimensional vector space. Let $\mathcal{E}$ denote the set of errors detectable by the error-correcting code $\mathcal{Q}$. We say that a linear operator $E$, which introduces noise to our input state, is detectable by the code $\mathcal{Q}$ if and only if 

\begin{equation}
\bra{i_L}E\ket{j_L} = C(E) \delta_{ij},
\end{equation}

where $\ket{i_L}$ and $\ket{j_L}$ are logical basis vectors of the code space $\mathcal{Q}$, and $C(E)$ is a constant that depends on the operator $E$. This condition is a property of error-correcting codes, also discussed and proved in \cite{Nielsen_Chuang_2010}. 

The concept of distance in error correction comes into play when we discuss local error operators, which are defined as

\begin{equation}
E = M_1 \otimes M_2 \otimes \cdots \otimes M_n,
\end{equation}
where $M_i$ are linear operators acting on individual qubits in the vector space $\mathbb{C}^2$. The weight of a local error operator is defined as the number of $M_i$'s that are not scalar multiples of the identity operator. We say a code has a distance at least $t$ if and only if all local error operators with weight less than $t$ are detectable by the error-correcting code $\mathcal{Q}$. Such codes are referred to as $(n, K, t)$ codes. These codes are considered pure if 

\begin{equation}
\bra{i_L}E\ket{i_L} = \frac{1}{2^n} \text{Tr}(E)
\end{equation}

for all $E \in \mathcal{E}$ and $\ket{i_L}$ belonging to the codes logical basis.
This is a high-level formalism of the error-correcting code $\mathcal{Q}$. A more detailed and concrete definition can be found in \cite{Scott_k-uniform}. Since we will not be explicitly using those definitions they have been not discussed in this article.

An interesting result, proven in \cite{Scott_k-uniform}, shows that a pure quantum state $\ket{\psi}$ is $k$-uniform if and only if $\ket{\psi}$ corresponds to a $(n, 1, k+1)$ error-correcting code, as discussed above. This property leads us to another application of our average distillable entanglement. 

\begin{theorem}
    A pure quantum state $\ket{\psi}$ has average distillable entanglement equal to $\frac{n}{2^n} \left(2^{n-1} - \binom{n-1}{\lfloor \frac{n}{2} \rfloor} \right)$ iff $\ket{\psi}$ belongs to the subspace of QEC $(n, 1, \lfloor \frac{n}{2} \rfloor+1)$.
\end{theorem}

\begin{proof}
    
The entanglement measure introduced in this work, the average distillable entanglement, reaches its maximum value (as stated in Proposition~\ref{prop:ADE_iff_AME}) if and only if the quantum state is AME. We know that the Absolutely maximally entangled (AME) states reside in the QEC subspace $(n,1,\lfloor \frac{n}{2} \rfloor+1)$ (ref \cite{Scott_k-uniform}). Therefore, the theorem follows.
\end{proof}

We have demonstrated that the average distillable entanglement measure is useful for both quantifying entanglement and identifying error-correcting codes.
\subsection{Integrated Euler Characteristic}
In this section, we explore how the average distillable entanglement is related to the Integrated Euler Characteristic (IEC) of an entanglement complex derived from the multipartite state.
\begin{theorem}

\label{theroem:euler_characteristic_distillation}
In
Let $\rho = \ketbra{\psi}{\psi}$ be a pure multipartite quantum state with qubits set $\mathcal{A} := \{\mathcal{A}_1, \dots \mathcal{A}_n\}$. The average distillable entanglement is bounded by the reduced Integrated Euler characteristic $\tilde{\mathfrak{X}}(\infty)$ of the persistence module created using functional $C(J)_\rho = \sum_{v \in J} S(v)_\rho - S(J)_\rho$ for $J \subseteq A$,
\begin{equation}
\frac{\tilde{\mathfrak{X}}(\infty)}{2^n}  \leq \langle E_D \rangle \leq \frac{\tilde{\mathfrak{X}}(\infty)}{2^n}  + \frac{1}{2} \varepsilon_{\text{max}}  
\end{equation}
where $\langle E_D \rangle$ is the average distillable entanglement of $\rho$ and $\varepsilon_{max}$ is the maximum filtration parameter where the barcodes of persistence module stop changes.
\end{theorem}
\begin{proof}
According to Theorem 4.18 in  \cite{hamilton2023probing}, the reduced integrated Euler characteristic is given as an alternating sum over subsets of the total Hilbert space as
\begin{flalign}
    \tilde{\mathfrak{X}}_q(\infty) &=  \sum_{J \subseteq  \mathcal{A}} (-1)^{|J| - 1} S_q(J)_{\rho}.
\end{flalign}
For pure states $\rho_{AB}$, Ref.\cite{bennet_distillation} showed that the entanglement distillation rate, when partitioning the state $\rho$ into subsystems $A$ and $B$, can be quantified by the entropy of entanglement, given by
\begin{equation}
\label{eq:Distillation_Cost_Equality}
     E_D(\rho_{A}) = S(\rho_A) = -\Tr(\rho_A \log \rho_A)
\end{equation}
where $\rho_A = \Tr_B(\ket{\psi}\bra{\psi})$. For reduced density state $\rho_A$, the von Neumann entropy is denoted as $S(A)_\rho \equiv S(\rho_A)$.

As $q \rightarrow 1$, Tsallis entropy defined as $S_q(J)_{\rho}$ approaches the von Neumann entropy, 
\begin{equation}
\lim_{q \rightarrow 1} S_q(J)_{\rho} = S(J)_{\rho} = -\Tr(\rho_J \log \rho_J)    
\end{equation}
Using the von Neumann entropy formalism in the previous equation, 
\begin{flalign}
   \tilde{\mathfrak{X}}(\infty) &=  \sum_{J \subseteq  \mathcal{A}} (-1)^{|J| - 1} S(J)_{\rho}  \\
    \label{eq:IEC_rearranged}
    \tilde{\mathfrak{X}}(\infty) &=  \sum_{\substack{J \subseteq \mathcal{A}}}  S(J)_{\rho} - 2 \sum_{\substack{J \subseteq \mathcal{A} \\ |J| \text{mod} 2 = 0}}  S(J)_{\rho} 
\end{flalign}
Dividing equation $\eqref{eq:IEC_rearranged}$ by $2^n$, we get
\begin{flalign}
    \frac{\tilde{\mathfrak{X}}(\infty)}{2^n} &= \frac{\sum_{\substack{J \subseteq \mathcal{A}}}  S(J)_{\rho}}{2^n} - \frac{1}{2^{n-1}}\sum_{\substack{J \subseteq \mathcal{A} \\ |J| \text{mod} 2 = 0}}  S(J)_{\rho} 
\end{flalign}
Substituting in equation \eqref{eq:Distillation_Cost_Equality} in our previous equation, we get
\begin{flalign}
    \frac{\tilde{\mathfrak{X}}(\infty)}{2^n} &= \frac{\sum_{\substack{J \subseteq \mathcal{A}} }E_D(\rho_{J})}{2^n} - \frac{1}{2^{n-1}} \sum_{\substack{J \subseteq \mathcal{A} \\ |J| \text{mod} 2 = 0}}  S(J)_{\rho}  \\
    \label{eq:IEC_avgDistillation_equality}
     &= \langle E_D \rangle - \frac{1}{2^{n-1}}\sum_{\substack{J \subseteq \mathcal{A} \\ |J| \text{mod} 2 = 0}}  S(J)_{\rho}   \\
     \frac{\tilde{\mathfrak{X}}(\infty)}{2^n} & \leq \langle E_D \rangle
\end{flalign}
as von Neuman Entropy is non-negative, ie $S(J)_\rho \geq 0$ for any $J \subseteq \mathcal{A}$.

The total correlation function for a multipartite quantum state $ \rho $ defined on the set $\mathcal{A} := \{ A_1, .... A_n \}$ is defined as 
\begin{equation}
    C(\mathcal{A})_{\rho} = \sum_{v \in \mathcal{A}} S(v)_{\rho} - S(\mathcal{A})_{\rho} 
\end{equation}
On rewriting the equation, the entropy of the system on set $\mathcal{A}$ can be written in terms of the entropy of individual subsystems and correlation function
\begin{equation}
    \label{eq:entropy_in_terms_of_correlation}
    S(\mathcal{A})_{\rho} = \sum_{v \in \mathcal{A}} S(v)_{\rho} -  C(\mathcal{A})_{\rho}
\end{equation}
Substituting in eq \eqref{eq:entropy_in_terms_of_correlation} into eq \eqref{eq:IEC_avgDistillation_equality} yields the following result
\begin{flalign}
    \frac{\tilde{\mathfrak{X}}(\infty)}{2^n} &= \langle E_D \rangle - \frac{1}{2^{n-1}}\sum_{\substack{J \subseteq \mathcal{A} \\ |J|\!\!\!\! \mod 2 = 0}} \left( \sum_{v \in J} S(v)_{\rho} -  C(J)_{\rho} \right) 
\end{flalign}
For an element $v \in A$, the number of even-length subsets that include element $v$ can be determined by counting the number of odd-length subsets of the remaining $n - 1$ elements. The number of odd-length subsets of a set with $n-1$ elements is given by: $\sum_{k \mod 2 = 1} \binom{n - 1}{k} = 2^{n-2}$.
\begin{flalign}
    \frac{\tilde{\mathfrak{X}}(\infty)}{2^n} &= \langle E_D \rangle - 2^{n-2} \cdot \frac{1}{2^{n-1}}\sum_{v \in \mathcal{A}}  S(v)_{\rho} + \frac{1}{2^{n-1}} \sum_{\substack{J \subseteq \mathcal{A} \\ |J| \!\!\!\!\mod 2 = 0}}  C(J)_{\rho} \nonumber\\
    &=\langle E_D \rangle - \frac{1}{2}\sum_{v \in \mathcal{A}}  S(v)_{\rho} + \frac{1}{2^{n-1}} \sum_{\substack{J \subseteq \mathcal{A} \\ |J| \!\!\!\!\mod 2 = 0}}  C(J)_{\rho}
\end{flalign}
\begin{equation}
      \langle E_D \rangle = \frac{\tilde{\mathfrak{X}}(\infty)}{2^n} + \frac{1}{2}\sum_{v \in \mathcal{A}}  S(v)_{\rho} - \frac{1}{2^{n-1}} \sum_{\substack{J \subseteq \mathcal{A} \\ |J| \!\!\!\!\mod 2 = 0}}  C(J)_{\rho}
\end{equation}
The correlation function is non-negative, a result that follows from the property of subadditivity. By omitting the correlation function from the previous equation, we obtain
\begin{flalign}
    \langle E_D \rangle &\leq \frac{\tilde{\mathfrak{X}}(\infty)}{2^n} + \frac{1}{2} \sum_{v \in \mathcal{A}}  S(v)_{\rho} 
\end{flalign}
Based on the definition provided in \cite{hamilton2023probing},  $\varepsilon_{max}$ is the largest value of the filtration parameter at which the number of barcodes changes. For any pure quantum state $ \rho $ defined on the set $ \mathcal{A} $, $ \varepsilon_{\text{max}} = \sum_{v \in \mathcal{A}} S(v)_{\rho} $. At this value of $ \varepsilon_{\text{max}} $, the entire topological structure forms a simplex, collapsing the barcode to a single 0-dimensional hole. The lower bound can be defined in terms of average distillable entanglement and $\varepsilon_{\text{max}}$,
\begin{equation}
\label{eq:ADE_upper_bound}
   \langle E_D \rangle   \leq   \frac{\tilde{\mathfrak{X}}(\infty)}{2^n} + \frac{1}{2} \varepsilon_{\text{max}} 
\end{equation}
\end{proof}


Computing the IEC provides us with a means of estimating an upper bound for the average distillable entanglement. For the odd-qubit case, we observe a particularly simple bound formula because of the symmetric properties of average distillable entanglement in such cases.


\begin{prop}
\label{prop:vanishing_IEC_for_odd_qubit}
    Let $\ket{\psi}$ be pure quantum state with odd number of qubits defined over qubit set set $\mathcal{A} := \{ A_1, .... A_n \}$  where $n = 2k + 1$ for some integer $k$. The Integrated Euler's characteristic for density matrix $\rho = \ket{\psi} \bra{\psi}$ is given by $\tilde{\mathfrak{X}}(\infty) = 0$.
\end{prop}
\begin{proof}

We defined the IEC of any system as the alternating sum of entropic measures over subsystems of increasing sizes. We have 
\begin{equation}
    \tilde{\mathfrak{X}}(\infty) =  \sum_{J \subseteq  \mathcal{A}} (-1)^{|J| - 1} S(J)_{\rho}
\end{equation}
Since $|\mathcal{A}| = 2k+1$ and the fact that the entropy of the global pure state is vanishing, we can reorder the terms as 
\begin{equation}
\label{eq:vanishing_IEC}
    \tilde{\mathfrak{X}}(\infty) =  \sum^{k}_{i=1} \sum_{\substack{J \subset  \mathcal{A} \\ |J| = i}} \left( S(J)_{\rho} - S(J^c)_{\rho} \right)
\end{equation}
where $J^c \coloneqq \mathcal{A} \setminus J$.
Using same the Schmidt decomposition argument as in \ref{prop:ADE_iff_AME}, we know that for any $ J \subset \mathcal{A} $, $ S(J)_{\rho} = S(J^c)_{\rho} $. Using this fact in equation \ref{eq:vanishing_IEC}, the terms summed over cancel out leaving us with $\tilde{\mathfrak{X}}(\infty) = 0$.
\end{proof}

\begin{theorem}
    Let $\ket{\psi}$ be a pure quantum state, represented by the density matrix density matrix $\rho = \ket{\psi}\bra{\psi}$, with an odd number of qubits defined over set $\mathcal{A} := \{\mathcal{A}_1, \dots \mathcal{A}_n\}$ where $n = 2k + 1$ for some integer $k$. We can bound the average distillable entanglement for $\ket{\psi}$ as below
    \begin{equation}
        \label{eq:ADE_bound_odd_qubits}
        \langle E_D \rangle \leq  \frac{1}{2} \sum_{v \in \mathcal{A}}  S(v)_{\rho} = \frac{1}{2}C(A)
    \end{equation}
\end{theorem}
\begin{proof}
Since $ \ket{\psi} $ is a pure quantum state, the entropy $ S(A)_{\rho} = 0 $. The correlation value of the state is given by

\begin{equation}
C(A)_{\rho} = \sum_{v \in \mathcal{A}} S(v)_{\rho} - S(A)_{\rho} = \sum_{v \in \mathcal{A}} S(v)_{\rho}.
\end{equation}
We have already proven that, for the case of odd qubits, the IEC in our TDA analysis is zero (see \ref{prop:vanishing_IEC_for_odd_qubit}). Substituting this result into \ref{eq:ADE_upper_bound}, we obtain the bound given in equation \ref{eq:ADE_bound_odd_qubits}. 
\end{proof}

The tightness of the bounds prsented in \ref{theroem:euler_characteristic_distillation} is very ambiguous. However, we can make an estimation of the slack of the bound \eqref{eq:ADE_upper_bound} for AME states.

\begin{theorem}
For AME states $ \rho = \ket{\psi} \bra{\psi} $ over qubit set $\mathcal{A} := \{\mathcal{A}_1, \dots \mathcal{A}_n\}$, the slack in the upper bound in the average distillable entanglement given by Theorem \ref{theroem:euler_characteristic_distillation}
$$
\left|\frac{\tilde{\mathfrak{X}}(\infty)}{2^n} - \langle E_D \rangle + \frac{1}{2} \varepsilon_{\text{max}}\right| \in \Theta(\sqrt{n})
$$
\end{theorem}
\begin{proof}
Since $ \ket{\psi} $ is $ \left\lfloor \frac{n}{2} \right\rfloor $-uniform, it is also $ 1 $-uniform. The filtration parameter at $ \varepsilon_{\text{max}} $ marks the point at which we stop our TDA analysis, and the entire topological space collapses into a single simplicial complex. For an input quantum state $ \rho $, the filtration parameter at which $ \mathcal{A} $ becomes a simplicial complex is given by:

\[
C(A) = S(\mathcal{A}_1) + S(\mathcal{A}_2) + \dots + S(\mathcal{A}_n) - S(A).
\]

For a pure $ 1 $-uniform state, this expression simplifies to $ C(A) = n $ as $S(\mathcal{A}_i) = 1$ for $i \in [n]$ and $S(A) = 0$. Thus, we conclude that $ \varepsilon_{\text{max}} = n $ for $ 1 $-uniform states.
    Let's first consider the slack of the bound for quantum states with an odd number of qubits. Using the results from proposition \ref{prop:ADE_iff_AME} and \ref{prop:vanishing_IEC_for_odd_qubit}, we know that for an odd-qubit AME, where $ n = 2k + 1 $ for some integer $ k $, we have that 
    \begin{flalign}
    \label{eq:IEC_slack_for_odd_qubits}
        \frac{\tilde{\mathfrak{X}}(\infty)}{2^n} - \langle E_D \rangle + \frac{1}{2} \varepsilon_{\text{max}} &= \frac{2k+1}{2} - \frac{2k+1}{2^{2k+1}} \left(2^{2k} - \binom{2k}{k} \right)\\
        &=\frac{2k+1}{2^{2k+1}}\binom{2k}{k} 
    \end{flalign}
    A well-known Stirling approximation for central binomial coefficients tells us that
\begin{equation}
    \frac{2^{2k}}{2\sqrt{k}}
\leq \binom{2k}{k} \leq \frac{2^{2k}}{\sqrt{\pi k}}
\end{equation}
Using this approximation in Equation \ref{eq:IEC_slack_for_odd_qubits}, we get that 
\begin{flalign}
\frac{2k+1}{2^{2k+1}} \frac{2^{2k}}{2\sqrt{k}}&\leq \frac{\tilde{\mathfrak{X}}(\infty)}{2^n} - \langle E_D \rangle + \frac{1}{2} \varepsilon_{\text{max}} \leq \frac{2k+1}{2^{2k+1}} \frac{2^{2k}}{\sqrt{\pi k}} \\
    \frac{\sqrt{k}}{2} + \frac{1}{4\sqrt{k}} &\leq  \frac{\tilde{\mathfrak{X}}(\infty)}{2^n} - \langle E_D \rangle + \frac{\varepsilon_{max}}{2} \leq  \frac{\sqrt{k}}{\sqrt{\pi}} + \frac{1}{2\sqrt{\pi k}} 
\end{flalign}
Through a crude approximation to the sub-dominant $O(1/\sqrt{k})$ terms in the bounds we then see that
\begin{flalign}
    \frac{\sqrt{k}}{2} &\leq  \frac{\tilde{\mathfrak{X}}(\infty)}{2^n} - \langle E_D \rangle + \frac{\varepsilon_{max}}{2}\leq  \frac{3\sqrt{k}}{2\sqrt{\pi}} 
\end{flalign}
Therefore, the slack scales as $ \Theta(\sqrt{n}) $ for odd qubit AME states. 

Even qubit quantum states do not share the same mathematical convenience as the odd qubit case. The IEC of such states does not simply vanish because the entropy of a subset of qubits and its complement share the same sign and do not cancel each other out. We begin with the IEC equation and apply the Schmidt decomposition argument previously applied in Proposition \ref{prop:ADE_iff_AME}, considering a quantum state with $n = 2k$ qubits, where $k$ is an integer. We get that 

\begin{equation}
    \tilde{\mathfrak{X}}(\infty) =  \sum_{J \subseteq  \mathcal{A}} (-1)^{|J| - 1} S(J)_{\rho} = 2\sum^{k-1}_{i=1} (-1)^{i-1} \sum_{\substack{J \subset  \mathcal{A} \\ |J| = i}} S(J)_{\rho} - (-1)^k\sum_{\substack{J \subset  \mathcal{A} \\ |J| = k}} S(J)_{\rho} 
\end{equation}
By assumption, the quantum state is AME so the IEC takes the value (see \ref{def:AME_state} for details of these calculations)
\begin{flalign}
    \tilde{\mathfrak{X}}(\infty) &= 2\sum^{k-1}_{i=1} (-1)^{i-1} \binom{2k}{i}i - (-1)^k \binom{2k}{k} k 
\end{flalign}
Using identity $\binom{n}{i}i = n \binom{n-1}{i-1}$, we get 
\begin{flalign}
\label{eq:even_qubit_AME_IEC}
     \tilde{\mathfrak{X}}(\infty) &= 4k\sum^{k-1}_{i=1} (-1)^{i-1} \binom{2k-1}{i-1} - (-1)^k \binom{2k}{k} k \\
     &= 4k\sum^{k-2}_{j=0} (-1)^{j} \binom{2k-1}{j} - (-1)^k \binom{2k}{k} k
\end{flalign}
Using Lemma \ref{lemma:binomial_equation} we get, 
\begin{equation}
    \tilde{\mathfrak{X}}(\infty) = 4k (-1)^k\binom{2k-2}{k-2} - (-1)^k \binom{2k}{k} k \\
\end{equation}
For the case where $ k $ is even, we use \ref{prop:ADE_iff_AME} to get
\begin{flalign}
    \frac{\tilde{\mathfrak{X}}(\infty)}{2^n} + \frac{1}{2} \sum_{v \in \mathcal{A}}  S(v)_{\rho} - \langle E_D \rangle &= \frac{1}{2^{2k}}\left[ 4k \binom{2k-2}{k-2} -  \binom{2k}{k} k  \right] + \frac{2k}{2} -  \frac{2k}{2^{2k}} \left[ 2^{2k-1} - \binom{2k-1}{k}\right] \\
    &= \frac{k}{2^{2k}} \left[ 4\binom{2k-2}{k-2} - \binom{2k}{k} + 2\binom{2k-1}{k}\right] + k - k \\
    &= \frac{k}{2^{2k}} \left[ 4\binom{2k-2}{k-2} - \binom{2k-1}{k-1} - \binom{2k-1}{k} + 2\binom{2k-1}{k}\right] \\
    &= \frac{k}{2^{2k}} \left[ 4\binom{2k-2}{k-2} - \binom{2k-1}{k} - \binom{2k-1}{k} + 2\binom{2k-1}{k}\right] \\
    &= \frac{4k}{2^{2k}} \binom{2k-2}{k-2} \\
    &= \frac{2k(k-1)}{2^{2k}(2k-1)} \binom{2k}{k}
\end{flalign}
Using the Stirling approximation again, we get
\begin{flalign}
    \frac{2k(k-1)}{2^{2k}(2k-1)}\frac{2^{2k}}{2\sqrt{k}} &\leq \frac{\tilde{\mathfrak{X}}(\infty)}{2^n} - \langle E_D \rangle + \frac{\varepsilon_{max}}{2} \leq \frac{2k(k-1)}{2^{2k}(2k-1)}\frac{2^{2k}}{\sqrt{\pi k}} \\ 
 \frac{\sqrt{k}(1 - 1/k)}{2 - 1/k} &\leq \frac{\tilde{\mathfrak{X}}(\infty)}{2^n} - \langle E_D \rangle + \frac{\varepsilon_{max}}{2}\leq  2\frac{\sqrt{k}}{\sqrt{\pi}}\left( \frac{1 - 1/k}{2 - 1/k} \right)
\end{flalign}

For the case when $k$ is even, slack scales $ \Theta(\sqrt{n}) $ \\
In the second case where $n/2$ is odd, we perform similar calculations to calculate the slack
\begin{flalign}
        \frac{\tilde{\mathfrak{X}}(\infty)}{2^n} + \frac{1}{2} \sum_{v \in \mathcal{A}}  S(v)_{\rho} - \langle E_D \rangle &= \frac{1}{2^{2k}}\left[ \binom{2k}{k} k  - 4k \binom{2k-2}{k-2}\right] + \frac{2k}{2} -  \frac{2k}{2^{2k}} \left[ 2^{2k-1} - \binom{2k-1}{k}\right] \\
        &= \frac{1}{2^{2k}}\left[ \binom{2k}{k} k  - 4k \binom{2k-2}{k-2}\right] + \frac{2k}{2} -  \frac{2k}{2^{2k}} \left[ 2^{2k-1} - \binom{2k-1}{k}\right] \\
        &= \frac{k}{2^{2k}}\left[ \binom{2k}{k} - 4\binom{2k - 2}{k - 2} + 2\binom{2k-1}{k}\right] + k - k \\
        &= \frac{k}{2^{2k}}\left[ \binom{2k}{k} - 4\binom{2k - 2}{k - 2} + \binom{2k}{k}\right] \\
        &= \frac{k}{2^{2k}} \binom{2k}{k}\left(2 - \frac{4k(k-1)}{2k(2k - 1)} \right) \\
        &= \frac{2k}{2^{2k}} \binom{2k}{k}\left( 1 - \frac{k - 1}{2k - 1}\right) \\
        &= \frac{2k^2}{2^{2k}(2k - 1)} \binom{2k}{k}
\end{flalign}
Approximating the central binomial coefficient we get that 
\begin{flalign}
    \frac{2k^2}{2^{2k}(2k - 1)}\frac{2^{2k}}{2\sqrt{k}} &\leq \frac{\tilde{\mathfrak{X}}(\infty)}{2^n} - \langle E_D \rangle + \frac{\varepsilon_{max}}{2}\leq \frac{2k^2}{2^{2k}(2k - 1)}\frac{2^{2k}}{\sqrt{\pi k}} \\ 
    \sqrt{k} \left( \frac{1/k}{2 - 1/k} \right) &\leq \frac{\tilde{\mathfrak{X}}(\infty)}{2^n} - \langle E_D \rangle + \frac{\varepsilon_{max}}{2}\leq \frac{2\sqrt{k}}{\sqrt{\pi}} \left( \frac{1/k}{2 - 1/k} \right)
\end{flalign}

The slack is bounded by $\Theta(\sqrt{n})$.
\end{proof}

We observe that the upper bound introduced in Theorem \ref{theroem:euler_characteristic_distillation} scales as $\Theta(\sqrt{n})$ for AME states when the ADE is saturated. In contrast, the lower bound does not exhibit the same behavior—it scales linearly with $n$. We expect the upper bound to be tighter and more informative than the lower bound.
\section{Barcodes of Graph States} \label{sec:graph_states}

To continue our application of TDA to multipartite entanglement, we turn to a previously developed class of states known as graph states. These states are represented using a graph $G = (V, E)$ and are a generalization of cluster states, making them fundamental primitives for Measurement Based Quantum Computing \cite{nielsen2006}. These states can exhibit highly complex, non-local entanglement, which, along with their relatively simple mathematical description, makes them ideal candidates for studying multipartite entanglement. 

Graph states can be defined in two equivalent ways, one via quantum circuits and the other using the stabilizer formalism. We will first give the circuit description. Given a graph $G$, which for the remainder of this section we assume to be connected and at least 4 vertices, assign a qubit to each vertex $v \in V$. We will interchangeably refer to a vertex as its associated qubit and vice-versa. Graph states are constructed by first applying a Hadamard gate to each vertex, followed by a controlled-Z gate for every edge $(u, v) \in E$, acting between the corresponding vertices $u$ and $v$.
Mathematically, graph states are defined as 
\begin{equation}
    \ket{G} \coloneqq \prod_{(u, v) \in E} CZ(u, v) H^{\otimes V} \ket{0}^{\otimes V}. \label{def:graph_state}
\end{equation}
The controlled-Z rotation can be replaced with other 2-qubit unitaries but we will not make use of this generalization.

One of the features of graph states that make them easy to work with analytically is that they are stabilizer states, a class of states first developed in error correction and has found traction in various other research fields. We will first give a brief overview of the stabilizer formalism, focused specifically on the quantities we will need to compute, before moving on to studying graph states in particular. Stabilizer states are typically defined using the Pauli group $\mathcal{P}_n$, where $n$ denotes the number of qubits, where we include an overall global phase of $\set{\pm 1, \pm i}$. An element of the Pauli group $K \in \mathcal{P}_n$ can act on state vectors by multiplication and a state $\ket{\psi}$ is said to be stabilized by $K$ if $K \ket{\psi} = \ket{\psi}$. We will typically work with an abelian subgroup of $S \subseteq \mathcal{P}_n$ that is defined by a set of generators $S = \langle K_1, \ldots, K_m \rangle$. The set of states that is stabilized by $S$ can be seen to be in the image of the projection operator $\Pi_S = \prod_{K \in S} \Pi_K = \prod_{K \in S} \frac{1}{2}(\identity + K)$, where $K$ is the Pauli operator (which is a reflection operator).

 To construct the stabilizer representation of graph states we define a stabilizer generator $K_u$ for each vertex $u \in V$ as
\begin{equation}
    K_u \coloneqq X_u \prod_{v \in N(u)} Z_{v}.
\end{equation}
Then $\ket{G}$ is defined to be the unique stabilizer state for the stabilizer subgroup $S_V$ generated by each of the single qubit generators
\begin{equation}
    S_V \coloneqq \langle K_u \rangle_{u \in V},
\end{equation}
and can be written $\ketbra{G}{G} = \Pi_{S_V}$. One of the requirements to be a stabilizer group is that the group must be abelian, and since each element is a Pauli string and squares to the identity we can convert the product to a sum
\begin{align}
    \Pi_{S_V} &= \frac{1}{2^{|V|}} \prod_{u \in V} (\identity + K_u) = \frac{1}{2^{|V|}} \sum_{b_1 = 0, \ldots, b_{|V|} = 0}^{b_1 = 1, \ldots, b_{|V|} = 1} K_1^{b_1} \ldots K_{|V|}^{b_{|V|}} = \frac{1}{2^{|V|}} \sum_{K \in S_{V}} K.
\end{align}
This expression as a sum is ultimately where we derive most of our utility of the stabilizer formalism from, as we are now able to compute linear operators, such as expectation values and partial traces, as a sum of the operator acting on individual Pauli strings.

We now have our most useful form for the density matrix associated with a graph state
\begin{equation}
    \rho_G \coloneqq \ketbra{G}{G} = \frac{1}{2^{|V|}} \sum_{K \in S_V} K.
\end{equation}
Before we get into computations that this formalism makes easy, we define the support of a stabilizer as the set of qubits (vertices) that the operator acts non-trivially on. For the generators, this is simply the union vertex of the generator with the neighbors, denoted as $N(u)$, and written as
\begin{equation}
    \supp(K_u) = \set{u} \cup N(u).
\end{equation}
Note that stabilizers in $S_V$ are uniquely defined by the product of generators that constitute the stabilizer, which is a one-to-one mapping with subsets of vertices. Since every vertex stabilizer has a Pauli $X$ factor, this means that any product of generators will have a Pauli $X$ on each vertex in which a generator is present, as the $Z$ factors from other generators will not be able to cancel the $X$ factors. The $Z$ factors from each vertex may or may not cancel, depending on the number of vertices that share this neighbor. For example, if we take the stabilizer $K = K_u K_v$ and $u, v$ share only one neighbor $w$ in common, then the support of $K$ will be $\set{u} \cup \set{v} \cup N(u) \cup N(v) \setminus w$. This is because $Z_w^2 = \identity$. 

This leads to defining the symmetric difference of neighborhood sets as 
\begin{equation}
    \triangle_N(U) \coloneqq \set{v \in V : | U \cap N(v) | \mod 2 = 1}.
\end{equation}
This allows us to write the support of an arbitrary graph state stabilizer $K_U$ as
\begin{equation}
    \supp(K_U) = \supp(K_{u_1} K_{u_2} \ldots K_{u_{|U|}}) = U \cup \triangle_N(U). \label{eq:stabilizer_support}
\end{equation}
This matches with the notion of \emph{local sets}, defined in \cite{hoyer2006graph} and explored further in \cite{claudet2024covering}, where a set is said to be \emph{local} if it is of the form in Eq. \eqref{eq:stabilizer_support}. They are able to show a number of properties of \emph{minimal} local sets, algorithms to find covers of the original graph by minimal local sets, and their relations to cut ranks. An interesting property of local sets is that they are invariant under local complementation. We further develop the notion of local sets, or the support of a vertex set, to relate it to the marginal entropies of a graph state. The invariance of local sets under local complementation provides an explanation for why the marginal entropies of graph states also remain invariant under local complementation.

Intuitively, a vertex $v$ is in the symmetric difference of the neighbors of $U \subset V$ if and only if it is neighbors with an \emph{odd} a number of vertices $u \in U$. This allows us to compute the support of a stabilizer operator as
\begin{equation}
    \supp(K_U) = \supp(K_{u_1} K_{u_2} \ldots K_{u_{|U|}}) = U \cup \triangle_N(U),
\end{equation}
as every vertex $v \in \triangle_N(U)$ will have a $Z$ operator acting on it from $K_U$ and every vertex $u \in U$ will have an $X$ operator. We can use this remarkable fact to compute traces of stabilizers, as the trace of any Pauli operator is zero and therefore the partial trace over a set $A$ of a Pauli string $K$ is zero unless $\supp(K) \subseteq A^C$. This lets us define
\begin{equation}
    S_U \coloneqq \set{K \in S_V : \supp(K) \subseteq U }.
\end{equation}
To determine if a stabilizer $K$ is in $S_U$ one needs to compute the partial trace $K \in S_U \iff \partrace{V \setminus U}{K} \neq 0$.

Now that we can compute restrictions of stabilizers to subsets of vertices we can easily compute restrictions of the graph state to subsets of vertices. This is given by the following Lemma.
\begin{lemma}
    Let $G=(V,E)$ be a graph and let $U\subset V$ then the partial trace of the graph state over the complement of $U$ is given by
    \begin{align}
        \rho_U &= \partrace{V \setminus U}{\rho_G} \\
        &= \frac{1}{2^{|V|}} \partrace{V \setminus U}{\sum_{K \in S_V} K} \\
        &= \frac{1}{2^{|V|}} \sum_{K \in S_V} \partrace{V \setminus U}{ K} \\
        &= \frac{1}{2^{|V|}} \sum_{K \in S_U} K
    \end{align}
    
\end{lemma}

In the previous sections, we discussed the applications and utility of a modified topological invariant, the Integrated Euler Characteristic (IEC). While the IEC provides valuable information, Betti numbers offer a more detailed and nuanced understanding of the entanglement structure in quantum states. In this section, we aim to establish a relationship between the lower-dimensional Betti numbers—generated through Topological Data Analysis (TDA)—and graph states. We focus on finding any particular graph states with a unique topological footprint. These distinctive topological footprints may serve as verifiers for particular classes of entanglement or structural properties within quantum systems. We specifically choose graph states for our analysis because their entanglement patterns are inherently encoded in their underlying graph structures


For this article, we will focus only on the lower-dimensional holes, specifically 0-dimensional and 1-dimensional holes as analyzing higher-dimensional holes is computationally expensive and analytically challenging. To recall, what the 0/1-dimensional represents readers may refer to section \ref{sec:prelims}. The number of 0-dimensional holes decreases if an edge is formed between 2 disconnected components. The 1-dimensional holes in a topological structure are created by sticking in a bunch of 1-dimensional simplices which corresponds to an edge between pairs of vertices. Both of these rely on the construction of edges within the topological structure. Based on how we construct our topological structure, an edge is created when the filtration parameter exceeds the correlation measure between the two vertices. Thus, to determine the birth times and the number of 0-dimensional and 1-dimensional holes, we only require the correlation measure between vertex pairs in the input graph state.

We put forth the following proposition, which in essence tells us that the entropy of a subset of qubits is inversely proportional to the size of the support set of that subset. This proposition works for all subsets of vertices of the graph states/qubits of the quantum state. Later, we will use this proposition extensively when formalism specific properties of a graph is necessary and sufficient to produce a entropy value/correlation value.

\begin{prop} \label{prop:calculating_linear_entropy_using_stablizer_support}
    The linear entropy  $S_2(A)_G $ of the reduced density operator $\rho^A_G$ associated with graph $ G = (V, E) $ and vertex set $ A \subseteq V$ can be calculated as:
    \begin{equation} \label{eq:calculating_linear_entropy_using_stablizer_support}
     S_2(A)_G = 1 - \frac{|\mathcal{S}_A|}{2^{|A|}}
    \end{equation}
    where $\mathcal{S}_A$ denotes the set of stabilizer operators for the graph 
$G$ whose support lies on vertices in set $A$.
    
\end{prop}

\begin{proof}
     Let $A \subseteq V$ be a subset of vertices of graph $G$. Let $B$ be the corresponding complement of set A, defined as $B := V\setminus A$. The reduced density operator for set A is denoted by $\rho^A_G := \Tr_B(\ket{G}\bra{G})$. 
     
    Using equation (114) in \cite{hein2006entanglement} we find that,
    \begin{equation}
    \label{eq:partialtrace_stabilizer}
        (\rho^A_G)^2 = \frac{|\mathcal{S}_A|}{2^{|A|}} \rho^A_G
    \end{equation}
    
    Using equation \eqref{eq:partialtrace_stabilizer}, linear entropy $S_2(A)_{G}$ can be calculated as
    \begin{equation}
        S_2(A)_{G} = 1 - \Tr{((\rho^A_G)^2)} = 1 -  \Tr\left( \frac{|\mathcal{S}_A|}{2^{|A|}} \rho^A_G \right) =1 - \frac{|\mathcal{S}_A|}{2^{|A|}}
    \end{equation}
\end{proof}

\begin{prop}
\label{prop:support_stabilizer_subset_of_group_of_stablizer_operator}

For any graph $ G = (V, E) $, let $ A \subseteq V $ be a subset of vertices of size $n$, ie $A = \{1, \dots n\}$. The set of stabilizer operators that are supported on the vertex set $ A $, denoted $ \mathcal{S}_A $, is a subset of the  group of stabilizer operators in set $A$, 
\begin{equation}
\label{eq:support_stabilizer_subset_of_group_of_stablizer_operator}
\mathcal{S}_A \subseteq \langle K_1, K_2, \dots K_n \rangle
\end{equation}
\end{prop}
\begin{proof}
   This proposition generalizes the result presented in Proposition 14 of \cite{hein2006entanglement}. We will just provide an informal intuition of the proof here. Any stabilizer operator for a vertex $v \in V \setminus A$ applies an $X$ operation on vertex $v$, which is not cancelled by operations from other stabilizer operators. Therefore, any stabilizer containing $K_v$ acts non-trivially on the vertex set outside $A$, meaning the stabilizer cannot be supported by the set $A$.
\end{proof}

The proposition \ref{prop:calculating_linear_entropy_using_stablizer_support} provides a means of inferring the graph structure from the marginal entropies of the input quantum state. This proposition also justifies our choice of using linear entropy for our analysis, as it offers a straightforward method for calculating the entropy of the various subsystems of the graph state. For our purposes, selecting one entropic measure over another does not appear to yield any additional insights during our analysis. Therefore, the choice of entropy measure may depend on the convenience of calculation or interpretation rather than its ability to provide new findings.

An interesting consequence of the proposition is that for any connected graph state the linear entropy of any vertex is $0.5$. This property is also mentioned in \cite{hein2006entanglement} which states that any pure connected graph state has a maximally mixed single qubit subsystem. Another way to interpret this is that any connected pure graph states are $1-$uniform quantum systems \ref{def:k-uniform}. We can also show this property using Proposition \ref{prop:calculating_linear_entropy_using_stablizer_support}. The support of any non-isolated vertex $v$, denoted $\mathcal{S}_v$, can only contain either the identity or the stabilizer operator $K_v$. Since vertex $v$ is non-isolated, the stabilizer $K_v$ applies a $Z$-operator to the neighboring vertices of $v$, which is non-empty. Therefore, the support of $K_v$ lies outside $v$, and $K_v \notin \mathcal{S}_v$. The identity element trivially belongs to $\mathcal{S}_v$. Substituting the size of $\mathcal{S}_v$ as 1 into equation \ref{eq:calculating_linear_entropy_using_stablizer_support}, we obtain the linear entropy of $v$ as $0.5$. This property of connected graph states is important as we will rely on these numbers when calculating the value of correlation functions.

\begin{prop}
\label{prop:conditions_for_belonging_to_support_stabilizer}

Assume a graph $G = (V, E)$ with no isolated vertices, let $(u, v)$ be a pair of vertices in $V$. We propose

\begin{enumerate}
    \item $ K_u \in \mathcal{S}_{u,v} $ if and only if $ (N(u) \setminus \{v\}) = \emptyset $.
    \item $ K_v \in \mathcal{S}_{u,v} $ if and only if $ (N(v) \setminus \{u\}) = \emptyset $.
    \item $ K_u K_v \in \mathcal{S}_{u,v} $ if and only if $ (N(u) \setminus \{v\}) = (N(v) \setminus \{u\}) $.
\end{enumerate}
 where $K_u$ and $K_v$ are the stabilizer operator corresponding to vertex $u$ and $v$ respectively, and $\mathcal{S}_{u,v}$ is the set of stabilizer operators whose support lies in the vertex set $\{u, v\}$. 

\end{prop}

\begin{proof}
The backward direction of each of these conditions can be proven trivially. We just have to make an observation that whenever the right-hand side condition holds true then the stabilizer operator acts trivially on set $V \setminus \{a, b\}$. It is could be seen for the cases $ (N(u) \setminus \{v\}) = \emptyset $, $ (N(v) \setminus \{u\}) = \emptyset $, and $ (N(u) \setminus \{v\}) = (N(v) \setminus \{u\}) $. In each case, the $ Z $-operation either cancels out to form the identity or is never applied to vertices outside $ \{a, b\} $.

Lets proof our claim that if $K_u \in \mathcal{S}_{u,v}$, then $ (N(u) \setminus \{v\}) = \emptyset $. This claim follows as for stabilizer operator $K_u$ can belong to $ \mathcal{S}_{u,v}$ only if $K_u$ acts trivially on set outside $\{a, b\}$. This condition is satisfied by the condition  $ (N(u) \setminus \{v\}) = \emptyset $ which says the neighbor set of vertex $u$ outside $\{u, v\}$ is empty. Similar argument holds true for $K_v \in \mathcal{S}_{u,v}$.

Finally, we need to prove that if $ K_u K_v \in \mathcal{S}_{u,v} $, then $ (N(u) \setminus \{v\}) = (N(v) \setminus \{u\}) $. This holds because the expressions $ (N(u) \setminus \{v\}) $ and $ (N(v) \setminus \{u\}) $ describe the sets of vertices outside $ \{u, v\} $ where the stabilizer operators $ K_u $ and $ K_v $ act non-trivially, respectively. If these sets match, the $ Z $-operations applied by each stabilizer operator cancel out. Consequently, $ K_u K_v $ acts trivially on vertices outside $ \{u, v\} $. If the sets do not match, then a $ Z $-operator must act on at least one vertex outside $ \{u, v\} $, implying that $ K_u K_v $ would not belong to $ \mathcal{S}_{u,v} $.

This can be proven more algebraically using the fact that $ K_u K_v \in \mathcal{S}_{u,v} $ if and only if $ N(u) \triangle N(v) \subseteq \{u, v\} $ or equivalently $(N(u) \setminus N(v)) \cup (N(v) \setminus N(u)) \subseteq \{u, v\}$. Another way to write this condition is that $(N(u) \setminus N(v)) \subseteq \{v\}$ and $(N(v) \setminus N(u)) \subseteq \{u\}$. The following equality holds for neighborhood set of vertex $v$, $N(u) = \{v\} \cup (N(u) \cap N(v))$ or $N(u) = N(u) \cap N(v)$. Similarly, $N(v) = \{u\} \cup (N(v) \cap N(u))$ or $N(u) = N(v) \cap N(u)$. Note that $u$ or $v$ cannot belong to $N(v) \cap N(u)$. Upon doing a set minus operation we get, $N(u) \setminus \{v\}= N(u) \cap N(v)$ and $N(v) \setminus \{u\}= N(u) \cap N(v)$. Thus we conclude that $ K_u K_v \in S_{u,v} $ if and only if $ (N(u) \setminus \{v\}) = (N(v) \setminus \{u\}) $.

\end{proof}
We build our topological structure using correlation measures of all possible subsets of vertices. The topological summaries would also depend on these measures. To relate the topological summaries to the entanglement structure of the graph, we would need to find the graph structural conditions for different values of the correlation function. These values of correlation measures depend on the size of the stabilizer state, so we would derive these conditions using the previous proposition. There are multiple ways a particular size of the support set could be achieved. Consequently, we often get an umbrella of conditions that could produce a particular correlation function value. For example, since there 3 ways to choose sets of order 2 from sets of order 3, we get 3 possible conditions that could lead to support of stabilizer being of size $2$.

\begin{prop} \label{prop:correlation_function_relation_with_graph_structure}
Given a graph $ G = (V, E) $ with no isolated vertices, we propose that the 2-deformed correlation function for any vertex pair $ (u, v) $, denoted $ C_2(u, v) $, satisfies the following conditions:
\begin{enumerate}
    \item $C_2(u, v) = 1 \iff (N(u) \setminus \{v\} = N(v) \setminus \{u\} = \emptyset)$
    \item
    $C_2(u, v) = 0.5 \iff 
\begin{aligned}
    & (N(u) \setminus \{v\} = N(v) \setminus \{u\} \neq \emptyset) \text{ or } \\
    & (N(u) \setminus \{v\} = \emptyset \text{ and } N(v) \setminus \{u\} \neq \emptyset) \text{ or } \\
    & (N(v) \setminus \{u\} = \emptyset \text{ and } N(u) \setminus \{v\} \neq \emptyset)
\end{aligned}$

    \item $ C_2(u, v) = 0.25 \iff (N(u) \setminus \{v\}) \neq (N(v) \setminus \{u\}) \text{ and } (N(u) \setminus \{v\}), (N(v) \setminus \{u\}) \neq \emptyset $

\end{enumerate}

where $ N(u) $ and $ N(v) $ represent the neighborhood sets of vertices $ u $ and $ v $, respectively.
\end{prop}
\begin{proof}
    
Assume a pair of vertices $u, v \in V$. As discussed earlier, we know that the linear entropy of vertices $u$ and $v$ is $0.5$. Plugging in these values in the equation for a 2-deformed correlation function, 

\begin{align}
    C_2(u, v) &= S_2(u) + S_2(v) - S_2(u, v) \\
    &=0.5 * 2 - \left(1 - \frac{|\mathcal{S}_{u,v}|}{2^{|\{u,v\}|}}\right) \\
\label{eq:correlation_function_relation_to_number_of_stabilizer_operator}
    &= \frac{|\mathcal{S}_{u,v}|}{4} 
\end{align}
\begin{itemize}
    \item \textbf{Condition 1}: The above equation that gives us the condition $ C_2(u, v) = 1 \iff |\mathcal{S}_{u,v}| = 4 $. This condition is satisfied if and only if $ \mathcal{S}_{u,v} = \{\identity, K_u, K_v, K_uK_v\} $ as $\mathcal{S}_{u,v} \subseteq \langle K_u, K_v \rangle$. This is only possible if all the conditions presented in proposition \ref{prop:conditions_for_belonging_to_support_stabilizer} are satisfied. Consequently, we have that $ N(u) \setminus \{v\} = \emptyset $ and $ N(v) \setminus \{u\} = \emptyset $, leading us to the conclusion

    \begin{equation}
    C_2(u, v) = 1 \iff (N(u) \setminus \{v\} = N(v) \setminus \{u\} = \emptyset).
    \end{equation}
    
    \item \textbf{Condition 2}: From equation \ref{eq:correlation_function_relation_to_number_of_stabilizer_operator} we get the equality that  $ C_2(u, v) = 0.5 \iff |\mathcal{S}_{u,v}| = 2 $. Since the identity operator is trivially always part of $\mathcal{S}_{u,v}$. The condition $|\mathcal{S}_{u,v}| = 2 $ is true if and only if only one of $K_u$ or $K_v$ or $K_{u, v}$ is part of $\mathcal{S}_{u,v}$. This enforces the constraint that the conditions presented in proposition \ref{prop:conditions_for_belonging_to_support_stabilizer} are exclusive. In other words, if one of the conditions is true then the rest of the conditions becomes false. 
    Lets look at the case when $K_u \in \mathcal{S}_{u,v}$ while $K_v \not \in \mathcal{S}_{u, v} \text{ and } K_{u, v} \not \in \mathcal{S}_{u, v}$. Using \ref{prop:conditions_for_belonging_to_support_stabilizer}, we get the condition that $(N(u) \setminus \{v\}) = \emptyset $ and $(N(v) \setminus \{u\}) \neq \emptyset $. Similar results is obtained when considering $K_v \in \mathcal{S}_{u,v}$ while $K_u \not \in \mathcal{S}_{u, v} \text{ and } K_{u, v} \not \in \mathcal{S}_{u, v}$. 
    Assume $K_uK_v \in \mathcal{S}_{u,v}$ and $K_u \not \in S_{u, v} \text{ and } K_v \not \in \mathcal{S}_{u, v}$. Using \ref{prop:conditions_for_belonging_to_support_stabilizer}, we get the conditions that $ (N(u) \setminus \{v\}) = (N(v) \setminus \{u\}) $ and $(N(u) \setminus \{v\}) \neq \emptyset$ and $(N(v) \setminus \{u\}) \neq \emptyset$. Summarizing the results we get that 
    
\begin{equation}
    C_2(u, v) = 0.5 \iff 
\begin{aligned}
    & (N(u) \setminus \{v\} = N(v) \setminus \{u\} \neq \emptyset) \text{ or } \\
    & (N(u) \setminus \{v\} = \emptyset \text{ and } N(v) \setminus \{u\} \neq \emptyset) \text{ or } \\
    & (N(v) \setminus \{u\} = \emptyset \text{ and } N(u) \setminus \{v\} \neq \emptyset)
\end{aligned}
\end{equation}

    \item \textbf{Condition 3}: From equation \ref{eq:correlation_function_relation_to_number_of_stabilizer_operator}, we can form condition that $C_2(u, v) = 0.25 \iff |\mathcal{S}_{u,v}| = 1$. Using the result presented in \cite{guhne2009entanglement}, we know that all stabilizer operators of the vertex set $\{u, v\}$ (except for the identity operator) have support on the vertex set outside $\{u, v\}$  if and only if $ (N(u) \setminus \{v\}) \neq (N(v) \setminus \{u\}) $ and both $ (N(u) \setminus \{v\}) $ and $ (N(v) \setminus \{u\}) $ are non-empty. Thus we arrive at the conclusion that 
    \begin{equation}
        C_2(u, v) = 0.25 \iff (N(u) \setminus \{v\}) \neq (N(v) \setminus \{u\}) \text{ and } (N(u) \setminus \{v\}), (N(v) \setminus \{u\}) \neq \emptyset
    \end{equation} 
\end{itemize}
\end{proof}

 These are all the propositions we would need to prove our main results. We propose that, by analyzing the entropic measures of every two-subset of qubits, we can detect graph states that are locally unitary (LU) equivalent to the GHZ state. This is a unique property of the GHZ state, as the entanglement across any bipartition remains constant and equal to 1 (when using von Neumann entropy). Due to the uniform entanglement structure, the conditions imposed by Proposition \ref{prop:correlation_function_relation_with_graph_structure} force the graph to adopt either a star or complete graph structure. Before presenting the theorem we will prove a lemma that would be useful while proving the main theorem.
 
 \begin{lemma}
\label{lemma:complete_graph_iff_neighbor_set_condition}
   For a graph G with no isolated vertices, Graph $G = (V, E)$ is a complete graph if and only if $N(u) \setminus \{v\} = N(v)\setminus \{u\} \neq \emptyset$ for every $u,v \in V$.
\end{lemma}
\begin{proof} 
Consider graph $G = (V, E)$ to be a complete graph, we will show that $N(u) \setminus \{v\} = N(v)\setminus \{u\}$ for every $u,v \in V$. Let $u, v$ be vertices in $G$, $u, v \in V$. By definition of a complete graph, every vertex pair in graph $G$ is connected by a unique edge. The neighborhood of $u$ and $v$ is $N(u) = V \setminus \{u\}$ and $N(v) = V \setminus \{v\}$ respectively. Computing $N(u) \setminus \{v\}$ and $N(v) \setminus \{u\}$ yields, 
\begin{equation}
N_u \setminus \{v\} = N_v \setminus \{u\} = V \setminus \{u, v\}    
\end{equation}
Hence, if the graph G is a complete graph, then $N(u) \setminus \{v\} = N(v)\setminus \{u \} \neq \emptyset$ for every $u,v \in V$.

Consider a graph $ G = (V, E) $ with no isolated vertex. Assume that $ N(u) \setminus \{v\} = N(v) \setminus \{u\} \neq \emptyset $ for any distinct vertex pair $ u, v \in V $. To reach a contradiction, assume that $G$ is not complete. Let $u, v \in V$ such that they are not adjacent and by assumption, we know $N(u) \setminus \{v\} = N(v)\setminus \{u\} \neq \emptyset$. This implies that there exists at least one vertex $w$ adjacent to both $u$ and $v$. We know that $N(u) \setminus \{w\} = N(w)\setminus \{u\} \neq \emptyset$, by assumption of the proposition. Since vertex $v$ is adjacent to vertex $w$ we have that $v \in  N(w)\setminus \{u\}$. This violates the condition that $N(u) \setminus \{w\} = N(w)\setminus \{u\} \neq \emptyset$ because vertex $u, v$ are non-adjacent. We get a contradiction.
\end{proof}

\begin{theorem}
\label{thm:graph_strcuture_iff_2-entropy}
Let $\rho = \ket{\psi}\bra{\psi}$ be a pure graph state with isolated vertices corresponding to graph $G = (V, E)$ in the Hilbert space $\mathcal{H_A} = (\mathbb{C}^2)^{\otimes n}$. We propose that $\rho$ is locally unitary equivalent to an $n$-qubit GHZ state if and only if the 2-deformed correlation measure of every 2-qubit subsystem of $\rho$ is 0.5, i.e., $C_2(J) = 0.5$ for all $J \subseteq A$.
\end{theorem}
\begin{proof}
    The forward direction of the proof is straightforward. Assume a quantum state $\rho$ such that it is $LU$-equivalent to the GHZ state. Notice that $\rho$ will share the marginal entropies as the GHZ state as $LU$ keeps the entropies invariant. 
The $ n $-qubit GHZ state is expressed as

\begin{equation}
    \ket{GHZ_n} = \frac{\ket{0}^{\otimes n} + \ket{1}^{\otimes n}}{\sqrt{2}}.
\end{equation}

When performing the partial trace over any subsystem of a GHZ state, the resulting reduced density matrix has a spectrum of $0.5$ with multiplicity $2$ and the remaining eigenvalues are 0. We can compute the linear entropy of any reduced density matrix $\rho_A$ as 
\begin{align}
\label{eq:2-deformed_correlation_value_1-uniform_state}
    S_2(A) &= 1 - \trace{\rho_A^2} \\
    &= 1 - \sum_{i} \lambda(i)^2 \\
    &= 1 - 2 \cdot \left(\frac{1}{2} \right)^2 = 1/2.
\end{align}
This holds for all proper subsets $ A \subset S $. The linear entropy of any 1 or 2-qubit subsystem is $ 0.5 $. Specifically, for any qubit pair $ u, v $, we have $ S_2(u) = S_2(v) = S_2(u, v) = 0.5 $. The 2-deformed correlation for any qubit subsystem is then computed as 
\begin{equation}
C_2(u, v) = S_2(u) + S_2(v) - S_2(u, v) = 2 \cdot 0.5 - 0.5 = 0.5.
\end{equation}

Let's prove the backward direction of the theorem. Let's assume that the 2-deformed correlation function for any pair of vertices $ (u, v) $ in $V$ is $ C_2(u, v) = 0.5 $. We will prove that the corresponding graph is a star or a complete graph.

There are 2 cases we consider here 
\begin{itemize}
    \item \textbf{Case 1 :} $\exists u \in V$ such that $deg(u) = 1$ \\
    Assume vertex $ u $ has degree 1, such that it is adjacent to $ w \in V $. For any vertex $v \in V\setminus\{u,w\}$ we assume that $C_2(u, v) = 0.5$. Using proposition \ref{prop:correlation_function_relation_with_graph_structure} and the fact that $N(u) \setminus \{v\} =\{w\} \neq \emptyset$ we can conclude that either $N(v) \setminus\{u\} = \emptyset$ or $N(v) \setminus\{u\} = \{w\}$. We see that condition that $N(v) \setminus\{u\} = \emptyset$ is not possible. The condition implies that either the neighbor set of vertex $v$ is empty or a leaf vertex adjacent to $u$. Since we assume our graph to be connected the neighbor set cannot be empty. On the other hand, vertex $v$ cannot be adjacent to $u$ as $u$ has degree 1 and is only adjacent to vertex $w$. We can arrive at the conclusion that for any vertex $v \in V\setminus\{u,w\}$, $N(v) = \{w\} $. This tells us that graph G is a star graph.
     \item \textbf{Case 2 :} $\forall u \in V$ such that $deg(u) > 1$ \\
In this case, we can apply Proposition \ref{prop:correlation_function_relation_with_graph_structure} to obtain the condition $ N(u) = N(v) \neq \emptyset $ for all $ u, v \in V $. By using Lemma \ref{lemma:complete_graph_iff_neighbor_set_condition}, we conclude that $ G $ must be a complete graph.
\end{itemize}
The graph state corresponding to the star or complete graph is LU equivalent to the GHZ state. Thus we conclude our proof.
\end{proof}

We can rewrite the above proposition with entropy measures instead of correlation measures. We choose to make this proposition with respect to correlation measures because it would be easier to relate to topological features which are constructed using these measures. Since we are assuming the graph G to be connected, the linear entropy of any single qubit subsystem is $0.5$. We get that the correlation measure for any pair of vertices $u, v$, $C(u, v) = 0.5 \iff S(u, v) = 0.5$. 

\begin{corollary}
    If a stabilizer state $\ket{\psi}$ has a 2-deformed correlation measure or linear entropy of every pair of qubits equal to $0.5$. Then the state is LU equivalent to the GHZ state.
\end{corollary}
\begin{proof}
    It is well known that every stabilizer state is LU equivalent to some graph state \cite{Walter_multipartite_entanglement}. Let $\ket{\psi}$ be a state that is locally unitary (LU) equivalent to the graph state $\ket{G}$, representing a connected graph $G$. Specifically, $\ket{\psi} = U_{\psi} \ket{G}$, where $U_{\psi}$ is a local unitary operation. A local unitary operation $U$ is more formally defined as $U = \otimes_{i=1}^{n} U_i$, where $U_i$ is a unitary operator acting on the $i$-th qubit, and $n$ is the number of qubits. The graph state $\ket{G}$ shares the same entanglement structure as $\ket{\psi}$. Given our assumptions regarding the entropy/correlation function values, along with Theorem $\ref{thm:graph_strcuture_iff_2-entropy}$, we can say that $\ket{G} = U_G \ket{GHZ}$, where $U_G$ is a local unitary operation. Consequently, we conclude that $\ket{\psi} = U_{\psi} U_G \ket{GHZ}$, which shows that $\ket{\psi}$ is LU equivalent to the GHZ state.
\end{proof}

 Given the marginal entropies of a graph state, there will be certain structural features of the graph state that produce such an entropic value. Every subset of vertices thus enforces a certain restriction of what the original graph can look like. If we are given the entropies of all the marginals, then constructing the original graph states becomes a constraint-satisfying problem. A set of marginal entropies may be consistent with a family of graph states because there are multiple ways a particular entropic value can be achieved by choosing different stabilizer operators to include in the support set.

 Let's look at how our topological structure denoted by $\mathcal{X} $ behaves for different values of the correlation function. This will help us build the intuition that leads to the next theorem, where we present the specific birth and death times of the Betti numbers for LU-equivalent GHZ states. The topological structure begins with $n$ disconnected vertices. At $\varepsilon = 0$, we will always have $n$ disconnected components and consequently $n$ $ 0-$dimensional holes. The number of $ 0-$dimensional holes decreases as more edges are introduced. When there is only one connected component, we are left with a single 0-dimensional hole, which persists indefinitely until we reach $ \varepsilon_{\text{max}} $, at which point the analysis stops. In our special case, the correlation function outputs 0.5 for any 2-subset of vertices in the topological structure $ \mathcal{X} $. We can be certain that no simplex is formed for $ \varepsilon < 0.5 $, as a 1-dimensional simplex cannot form within this range, and higher-dimensional simplices do not form before 1-dimensional simplices due to the downward closure property (see Ref. \ref{eq:downward_closure_total_correlation}). Consequently, the topological structure has $n$ $0-$dimensional holes for filtration parameters less than $0.5$. Conversely, if the TDA output shows $ n $ 0-dimensional simplices within the range $[0, 0.5)$, we can conclude that any pair of vertices has a correlation of at least 0.5. If the correlation function had taken a value within this range, an edge would have been formed, which would have led to a reduction in the number of 0-dimensional holes.

Now, let's examine the 1-dimensional holes. These holes correspond to the number of independent cycles in the topological structure $ \mathcal{X} $. It is important to emphasize the term ``independent" because if a cycle can be expressed as a linear combination of other cycles, it is not counted in the Betti number. Given our assumption that the graph state corresponds to a connected graph, the earliest filtration parameter at which a 2-simplex can form is $ 1 $ (see \ref{prop:calculating_linear_entropy_using_stablizer_support}). We can be confident that, in this case, there will be no higher-dimensional simplices that reduce the number of 1-dimensional holes. Therefore, the problem of finding 1-dimensional holes in the topological structure $ \mathcal{X} $ reduces to identifying the number of independent cycles in the graph. In the following theorem, we will present an explicit formula for computing this value.

\begin{theorem}
\label{thm:betti_number_iff_2-entropy}
There are $ n $ 0-dimensional holes for $ 0 \leq \varepsilon < 0.5 $. At $ \varepsilon = 0.5 $, one 0-dimensional hole and $ \frac{(n-1)(n-2)}{2} $ 1-dimensional holes form if and only if the graph state $\rho$ corresponding to graph $G =(V, E)$ has 2-deformed correlation measure of every 2-qubit subsystem of $\rho$ is 0.5, i.e., $C_2(J) = 0.5$ for all $J \subseteq A$.
\end{theorem}
\begin{proof}
Let's prove the forward direction first. Assume that for $ 0 \leq \varepsilon < 0.5 $, there are $ n $ 0-dimensional holes, and at $ \varepsilon = 0.5 $, there is one 0-dimensional hole and $ \frac{(n-1)(n-2)}{2} $ 1-dimensional holes.

Based on the lifespan of $0-$dimensional holes, we can infer that there exists no pair of vertices with a correlation measure less than $0.5$. At $\varepsilon = 0.5$. There is exactly 1 connected component ($ c = 1 $) corresponding to the single 0-dimensional hole. At this value of the filtration parameter, the space consists of $ n $ vertices, 1 connected component, and $ \frac{(n-1)(n-2)}{2} $ 1-dimensional holes, which correspond to the dimension of the cycle basis.

The dimension of the cycle basis is given in \cite{hage1996island} and is computed using the number of edges, vertices, and connected components
\begin{equation}
\text{Cycle basis dimension} = m - n + c,
\end{equation}
where $ m $ is the number of edges, $ n $ is the number of vertices, and $ c = 1 $ is the number of connected components. Substituting the values:
\begin{equation}
\frac{(n-1)(n-2)}{2} = m - n + 1
\end{equation}
Solving for $ m $, we find:
\begin{equation}
m = \frac{n(n-1)}{2}
\end{equation}
Therefore, the topological structure $ \mathcal{X} $ is a complete graph at $ \varepsilon = 0.5 $. As a result, the correlation measure between every pair of vertices is exactly 0.5, since for any pair of vertices $ u, v $, the condition $ C(u, v) \leq 0.5 $ is satisfied, and we also know that $ C_2(u, v) \geq 0.5 $, as no edges are within the range $[0, 0.5)$

Let's prove the backward direction of the statement. Assume that for every pair of vertices the $C_2(u, v) = 0.5$. Tracking the behaviour of the topological structure we see that there will be $n$ disconnected vertices in $\mathcal{X}$ and at $\varepsilon=0.5$ the entire space forms a complete graph. This explains the behavior of the $0-$dimensional holes. At $\varepsilon = 0.5$ we start to observe the birth of $1-$dimensional holes. The number of 1-dimensional holes corresponds to the dimension of the cycle basis. For a complete graph with $ n $ vertices, the number of edges is $ \frac{n(n-1)}{2} $, and there is 1 connected component. 


Substituting the values in the equation for the dimension of the cycle basis, we get:

\begin{equation}
\frac{n(n-1)}{2} - n + 1 = \frac{(n-1)(n-2)}{2}
\end{equation}

Thus, at $ \varepsilon = 0.5 $, there is a single 0-dimensional hole and $ \frac{(n-1)(n-2)}{2} $ 1-dimensional holes. Hence, the proposition follows in the backward direction.
\end{proof}

\begin{figure} [t!]
\label{fig:GHZ-barcodes}
    \centering
    \begin{minipage}{0.35\linewidth}
        \centering
        \includegraphics[height=5cm]{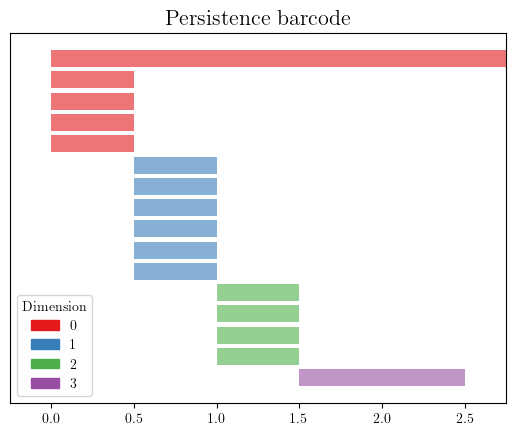}
    \end{minipage} 
    \hfill
    \begin{minipage}{0.45\linewidth}
        \centering
        \includegraphics[height=5cm]{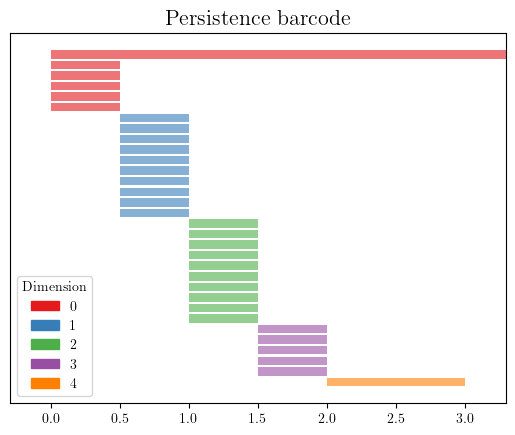}
    \end{minipage}
    \caption{Barcodes for 5- and 6-qubit GHZ states corresponding to the complete graph or star graph. In the 5-qubit case, all five $0$-dimensional holes vanish when the filtration parameter reaches $0.5$, at which point six $1$-dimensional holes emerge, leaving a single connected component. Similarly, in the 6-qubit case, all six $0$-dimensional holes disappear at the filtration value $0.5$, and ten $1$-dimensional holes are born at that same threshold. The barcodes here verify our result shown in Theorem \ref{thm:GHZ_state_iff_topological_footprint}.}

\end{figure}

\begin{theorem}
\label{thm:GHZ_state_iff_topological_footprint}
    Let $\rho = \ket{\psi}\bra{\psi}$ be a pure multipartite graph quantum state in Hilbert space $(\mathbb{C}^2)^{\otimes n}$. The quantum system $\rho$ generates $n$ 0-dimensional holes for $0 \leq \varepsilon < 0.5$ and at $\varepsilon = 0.5$, one 0-dimensional holes and $\frac{(n-1)(n-2)}{2}$ 1-dimensional hole if and only if the graph associated with $\rho$ is a star or complete graph.
\end{theorem}
\begin{proof}
   The proof follows directly from applying Theorem \ref{thm:betti_number_iff_2-entropy} and Theorem \ref{thm:graph_strcuture_iff_2-entropy}. Both the left-hand side (LHS) and right-hand side (RHS) of the theorem describe a scenario where the correlation function takes the value 0.5 for every pair of qubits in the graph state. Therefore, their equivalence follows.
 \end{proof}

We have shown that any graph state whose reduced density matrices (RDMs) for every two-qubit subset have eigenvalues ${1/2, 1/2}$ must be a GHZ state. This entanglement structure can be written more precisely: each two-qubit RDM must be locally equivalent to
\begin{equation} \label{eq:GHZ_condition_for_RDM} \frac{1}{2}\ket{00}\bra{00} + \frac{1}{2}\ket{11}\bra{11} \end{equation}
up to local unitaries. A natural question that arises is whether this result holds for all quantum states (not just graph states) with such entanglement properties. This is a variant of the marginal reconstruction problem which has been discussed in the past \cite{Linden_2002_2-Particle_RDM}, \cite{Walck_2008_GHZ_RDM}, \cite{Di_2004_RDM}. In the case of three qubits, this characterization holds: every 3-qubit 1-uniform state is a GHZ state \cite{Karol_k-uniform_construction}. Interestingly, 1-uniformity is a strictly weaker condition than requiring all two-qubit RDMs to take the form of Equation \eqref{eq:GHZ_condition_for_RDM}. Taking the partial trace of this RDM yields the maximally mixed state, consistent with 1-uniformity. However, we were unable to extend this result to systems with more than three qubits. A peculiar property of quantum states whose two-qubit RDMs take the form of Equation \eqref{eq:GHZ_condition_for_RDM} is that they remain invariant under SWAP tests. Such symmetric properties have also been discussed in previous works relating to $k-$uniform states \cite{Karol_k-uniform_construction}. We expect the symmetric property of the quantum state in addition to the eigenvalues of the 2-qubit RDM to be sufficient to proving this hypothesis.

\begin{figure}[t!]
    \centering

    \begin{subfigure}{0.43\linewidth}
        \centering
        \includegraphics[height=3.8cm]{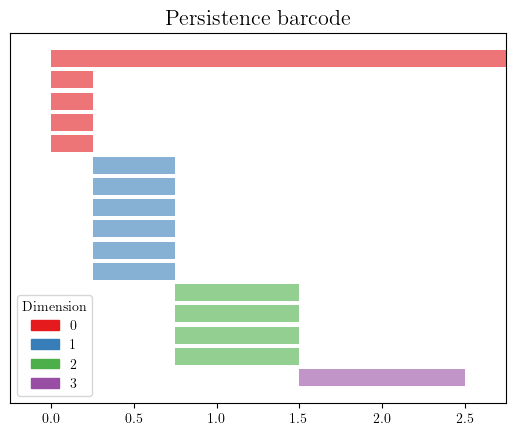}
        \caption{Persistence barcode for 5-qubit AME state}
    \end{subfigure}
    \hfill
    \begin{subfigure}{0.43\linewidth}
        \centering
        \includegraphics[height=3.8cm]{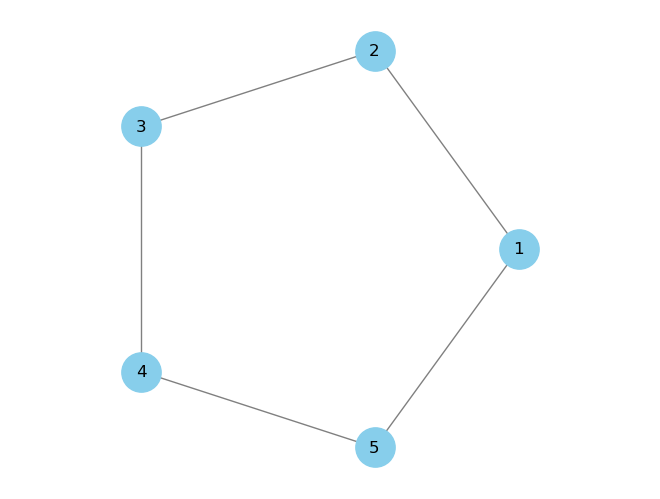}
        \caption{Graph representation for a 5-qubit AME state}
    \end{subfigure}

    \vspace{0.8em} 

    \begin{subfigure}{0.43\linewidth}
        \centering
        \includegraphics[height=3.8cm]{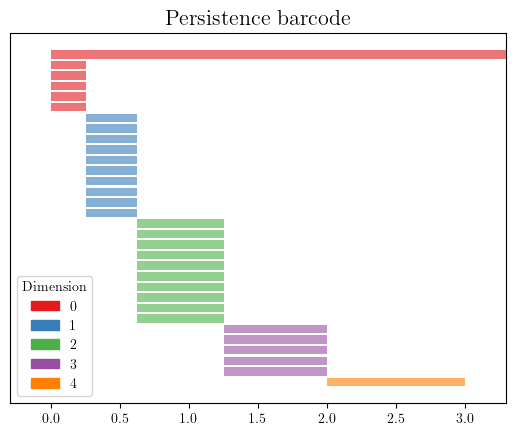}
        \caption{Persistence barcode for 6-qubit AME state}
    \end{subfigure}
    \hfill
    \begin{subfigure}{0.43\linewidth}
        \centering
        \includegraphics[height=3.8cm]{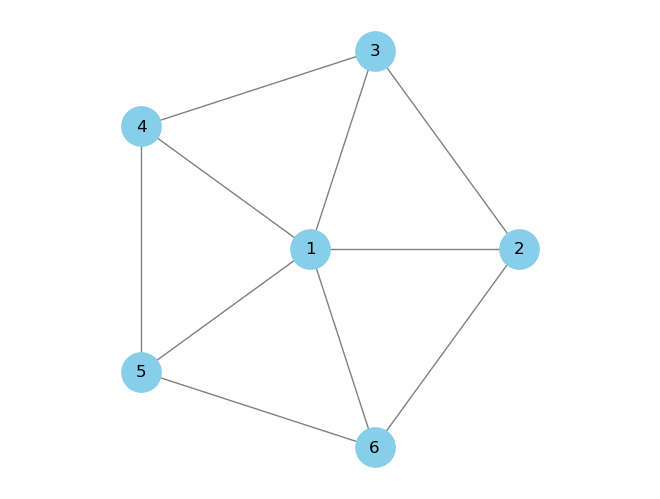}
        \caption{Graph representation for a 6-qubit AME state}
    \end{subfigure}

    \caption{Persistence barcodes and graph representations for 5- and 6-qubit AME states. The graph states corresponding to the pentagon and wheel graphs serve as examples of AME states, introduced in \cite{helwig2013quditgraph}. The birth times of the topological features (holes) align with the values discussed below.}
    \label{fig:AME_barcodes}
\end{figure}

\section{Barcodes of k-uniform/AME states}
\label{sec:k-uniform-barcodes}
We will now briefly discuss the barcode structure associated with $k$-uniform and Absolutely Maximally Entangled (AME) states. These quantum states exhibit share a very predictable barcodes pattern, which we seem worth discussing in this section. We will be using the 2-deformed correlation function \ref{eq:multipartite_2-deformed_total_correlation} to form our simplices. Using the total correlation function, as described in Section \ref{sec:distillable_entanglement}, is not suitable for these states. This is because all $\lfloor \frac{n}{2} \rfloor$-simplices, along with their lower-dimensional faces, appear immediately as soon the TDA analysis begins— (see Lemma \ref{lemma:AME_vanishing_correltion}). All homology classes in dimensions less than $\left\lfloor \frac{n}{2} \right\rfloor$ vanish immediately, as the appearance of $\left\lfloor \frac{n}{2} \right\rfloor$-dimensional simplices fills in the lower-dimensional holes.

For AME states in particular, all simplices emerge at a specific value of the filtration parameter $\varepsilon$, since all subsets of qubits with the same cardinality share the same reduced density matrices—specifically, the maximally mixed state. At particular epsilon values, say $\varepsilon_k$, the topological space gets completely covered with $\binom{n}{k}$ $k-$dimensional simplices. The number of holes formed seems be maximal for any topological structures with $n$ vertices. We can't prove this rigorously, but we will give an intuitive explanation as to why we think this is true. The $k^{th}$ betti number is given by the formula (for a more detailed explanation see \cite{berry2023topology}, also talked about in \ref{sec:prelims})
\begin{equation}
    \beta_k = \dim(\ker \partial_k)- \dim(\im \partial_{k + 1})
\end{equation}
The term $\dim(\operatorname{im} \partial_{k+1})$ remains unchanged when a $k$-simplex is removed. Upon removing a simplex from the chain complex $X^{(k)}$, it may either belong to the kernel of the boundary operator $\partial_k$ or not. If the simplex contributes to $\ker \partial_k$, then its removal decreases $\dim(\ker \partial_k)$, leading to a reduction in the $k$-th Betti number. Otherwise, the Betti number remains unchanged. The $k$-th Betti number in the chain complex $X^{(k)}$ seems be a non-increasing function with respect to the number of $k$-dimensional simplices. That is, as $k$-simplices are removed, the $k$-th Betti number does not increase. Any other chain complex $Y^{(k)}$ can be formed by removing $k-$simplex from chain complex $X^{(k)}$ that has $\binom{n}{k}$ $k-$dimensional simplices. The presence of higher-dimensional $(k+1)$-simplices can only reduce the $k$-th Betti number, as they fill in $k$-dimensional holes. The $k^{th}$ betti number of $Y^{(k)}$ should be less than or equal to that of $X^{(k)}$.

Every AME state are $1-$uniform state. The linear entropy of every individual qubit of the quantum state is equal to $0.5$, refer to \ref{eq:2-deformed_correlation_value_1-uniform_state}. Any subset of vertices $J$ such that $|J| = k \leq \lfloor\frac{n}{2} \rfloor$, the linear of entropy for qubits $J$ is 
\begin{figure}[t!]
    \centering
    \begin{subfigure}{0.3\linewidth}
        \centering
        \includegraphics[width=\linewidth]{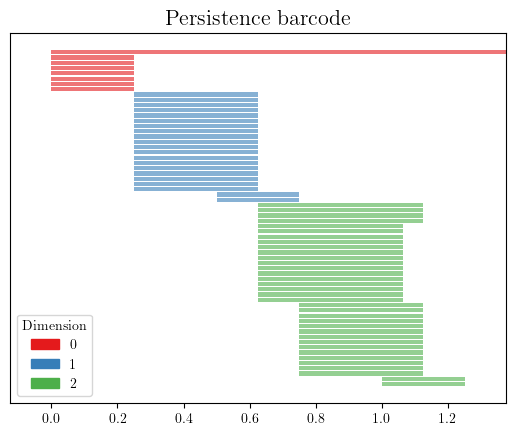}
    \end{subfigure}
    \hfill
    \begin{subfigure}{0.3\linewidth}
        \centering
        \includegraphics[width=\linewidth]{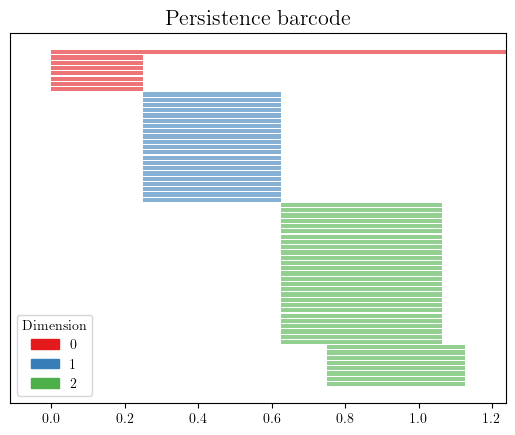}
    \end{subfigure}
    \hfill
    \begin{subfigure}{0.3\linewidth}
        \centering
        \includegraphics[width=\linewidth]{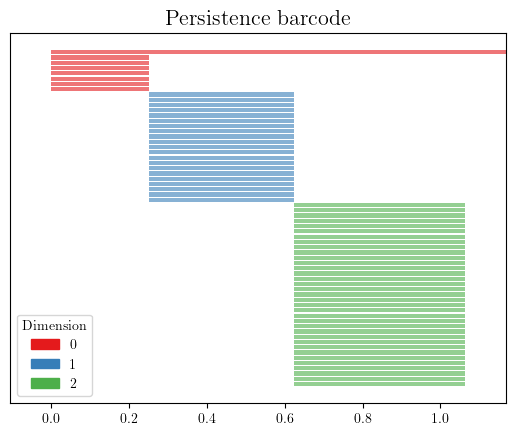}
    \end{subfigure}

    \vspace{1em} 

    \begin{subfigure}{0.3\linewidth}
        \centering
        \includegraphics[width=\linewidth]{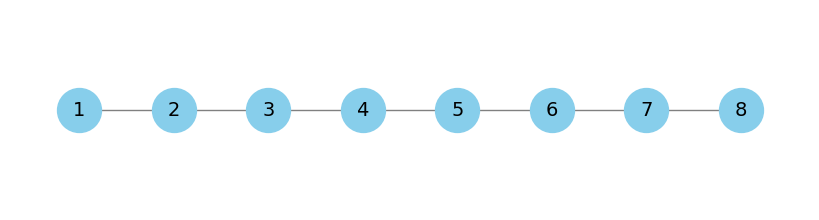}
        \caption{1-uniform state}
    \end{subfigure}
    \hfill
    \begin{subfigure}{0.3\linewidth}
        \centering
        \includegraphics[width=\linewidth]{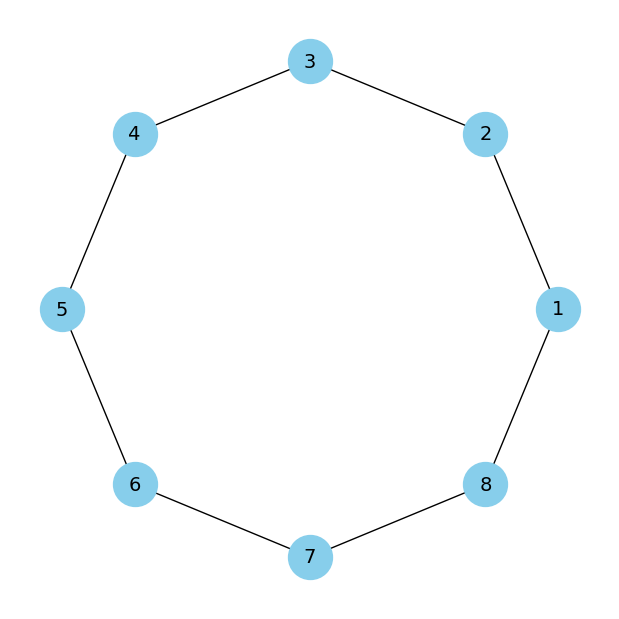}
        \caption{2-uniform state}
    \end{subfigure}
    \hfill
    \begin{subfigure}{0.3\linewidth}
        \centering
        \includegraphics[width=\linewidth]{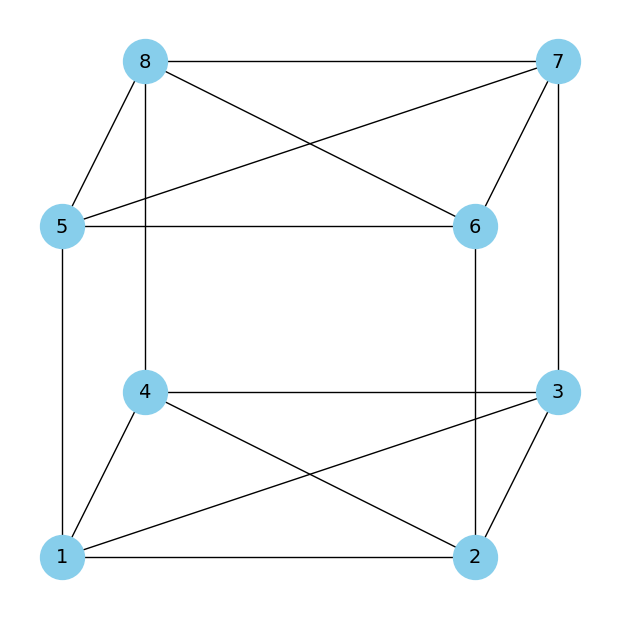}
        \caption{3-uniform state}
    \end{subfigure}

    \caption{Each column shows the persistence barcode (top) and corresponding graph (bottom) of an 8-qubit $k$-uniform quantum state for $k = 1, 2, 3$. The examples of 8 qubit $1, 2, 3$ uniform states have been borrowed from \cite{sudevan2022kuniformgraphs}. We observe that for $3$-uniform states, holes of dimensions $0$, $1$, and $2$ stack up; for $2$-uniform states, holes of dimensions $0$ and $1$ stack; and for $1$-uniform states, only $0$-dimensional holes stack.}
    \label{fig:k-uniform}. 
\end{figure}
\begin{align}
    S_2(J) &= 1 - \trace{\rho_J^2} \\
    &= 1 - \sum_{i=1}^{2^k} \lambda(i)^2 \\
    &= 1 - \frac{2^k}{(2^k)^2} = 1 - \frac{1}{2^k}
\end{align}
Similarly, in the case where $|J| = k > \left\lfloor \frac{n}{2} \right\rfloor$, we find that
\begin{equation}
    S_2(J) = 1 - \frac{1}{2^{n - k}}
\end{equation}
 The eigenvalues of the RDM for a subset of qubits $J$ can be verified by the properties of AME state \ref{def:AME_state}. 
Using these values, we get the 2-deformed correlation value of any subset of qubits $J = \{J_1, J_2,  \ldots J_k\}$ for $k \leq \lfloor\frac{n}{2} \rfloor$ can be calculated as 
\begin{align}
\label{eq:birth_time_AME_state}
    C_2(J_1, J_2,  \ldots J_k) &= \sum_{i=1}^{k} S_2(J_i) - S_2(J) \\ 
    &= \frac{k}{2}  + \frac{1}{2^k} - 1
\end{align}
Every $d-$dimensional hole is formed when the $d$-dimensional simplices are constructed. These holes appear at the correlation function values corresponding to subsets of $d+1$ qubits
\begin{equation}
\label{eq:birth_time_of_low_dim_simplex_AME}
    \varepsilon =  \frac{d+1}{2}  + \frac{1}{2^{d+1}} - 1 
\end{equation}
 for $d+1 \leq \lfloor\frac{n}{2} \rfloor$. For $d+1 > \lfloor \frac{n}{2} \rfloor$, we get a similar formula for the generation of $d-$dimensional holes we get that
 \begin{equation}
 \label{eq:birth_time_of_high_dim_simplex_AME}
 \varepsilon = \frac{d+1}{2}  + \frac{1}{2^{n - d - 1}} - 1    
 \end{equation}
 
 Therefore for AME states, we can predict the birth time of every $d-$dimensional hole. These formulas for epsilon are monotonic with respect to the simplex dimension. As a result, no higher-dimensional hole/simplex can appear before all lower-dimensional holes/simplices have formed and vanished. 
 
For $k$-uniform states, we observe a structure similar to that of AME states. Specifically, holes of dimension less than $k$, as well as those of dimension greater than  $n - k $, follow the birth time described by equations \eqref{eq:birth_time_of_low_dim_simplex_AME} and \eqref{eq:birth_time_of_high_dim_simplex_AME}. This pattern can be observed and verified in figures \ref{fig:k-uniform}. The persistence of the holes of dimensional $d$ such that $k \leq d \leq n-k$ is limited by the death time of $k$-simplices and birth time of $n-k$ simplices. The barcodes give us some idea of how the entanglement in $k-$uniform states must be distributed.

 Similar to GHZ states, $k$-uniform form \emph{all} $d-$dimensional simplices at a particular filtration parameter. For both GHZ and $k-$uniform state, the $d$-dimensional barcodes stack up and share the same birth and death time. The entropy of GHZ state across partition of size $n$ (minimum size of the two parts in a bipartition) scales as $\Theta(1)$ while the entropy for $k-$uniform state scales as $\Theta(n)$. Therefore, the birth and death time of GHZ state barcodes do not depend on the size of the state. 

\section{Conclusion}

Persistence Homology is a recently introduced TDA framework that examines the persistence of topological invariants making it useful for studying the geometric structure's underlying complex data. It puts forth lots of avenues for further study. In particular, the betti number, which quantifies the presence of d-dimensional holes in a space, lacks clear interpretation. Such invariants are interesting because they represent a peculiar property of the geometric structure such that on pulling and stretching of the structure they remain invariants. By analyzing the persistence of the holes, we have a way to distinguish between more significant and less significant topological features. Thus,  allowing us to focus on the most relevant ``holes" in the structure.

In this paper, we explore the application of topology to understand the entanglement structure and correlation distribution in quantum states in this paper. We put forth two main results. In Section \ref{sec:distillable_entanglement}, we investigate further the applications of the topological invariant, the Integrated Euler characteristic. The work of Hamilton and Leditzky showed that the Integrated Euler Characteristic is related to the $n$-tangle, but did not give an interpretation that clearly revealed its operational significance for quantum information processing tasks. We show that the integrated Euler Characteristic can be used to bound average distillable entanglement, which in turn measures the resourcefulness of a state for teleportation between the different parties sharing a multipartite state.  Additionally, we examine the utility of this entanglement measure in relation to $k$-uniform states and error correction codes. As a result, we observe that the topology of the state plays a role in determining its suitability for being a $k$-uniform state or an error-correcting code.
Our work therefore solidifies the use of the Integrated Euler Characteristic as a useful measure of multipartite entanglement beyond $3$ qubits~\cite{horodecki2024multipartiteentanglement}.

Our second result is motivated by the goal of finding an operational meaning and interpretation for the lower dimensional Betti number.  We show that the persistence of the 1-Betti number can   determine whether a given graph state is LU-equivalent to the GHZ state. We show that TDA can be employed to identify the GHZ state by analyzing the 0/1 Betti numbers for filtration parameters $ \varepsilon = 0.5 $ and $ \varepsilon \to 0.5^- $ (see Theorem \ref{thm:GHZ_state_iff_topological_footprint}). This verifier applies to \emph{all} quantum states in the case where the number of qubits is three.  This shows that low Betti numbers do infact play a useful role in multipartite entanglement and that the properties of persistence diagrams can be used to uniquely identify certain graph states.


There are a number of open questions that remain after this work.
More work is needed in order to understand the full meaning behind the $1$-Betti number. One promising direction is to interpret the invariant in terms of communication and the fidelity of state transfer. The 1-dimensional holes correspond to cycles that start and end at the same vertex. At a particular filtration parameter $\varepsilon$, if a cycle exists, we can interpret this cycle as representing the ability of an observer within the cycle to send a quantum state with a minimum level of fidelity, depending on $\varepsilon$, to any other observer in the cycle. However, the filtration parameter $ \varepsilon $ cannot provide information about the \emph{minimum} fidelity of state transfer. This limitation arises from the way we construct our topological structure. In the topological structure , a cycle exists if the correlation between vertices is at most $ \varepsilon $. Only the complement of the cycle will reveal the vertex pairs whose correlation exceeds a particular $ \varepsilon $. It would be interesting to find an operational interpretation of this in the context of communication protocols.

The current methodology for constructing the topological structure may seem counterintuitive. Topological summaries, such as the Betti number, are usually used to reveal potential noise in the dataset. The persistence of the invariant indicates the significance of the noise present in the dataset. Typically, we would expect noise to correspond to the lack of correlation between vertices; however, in our case, it is rather the opposite: the entanglement between the vertices. In our construction, if two vertices are highly entangled, they are placed far apart from each other. As a result, a 0-dimensional hole will persist for the duration of their correlation. Conversely, two uncorrelated vertices are placed closer together, and as a result, no 0-dimensional Betti number exists to identify this lack of correlation. We also explored additional measures that take their maximum value when the state is uncorrelated and the minimum value when it is highly correlated. However, we found that these measures do not necessarily satisfy the downward closure property, which is essential for simplicial complexes to be valid topological structures. 
It is an open question if more general hypergraphs that do not satisfy the downward closure property of simplicial complexes can be constructed and analyzed topologically that place correlated vertices closer together than uncorrelated ones. 

Perhaps the main question that remains involves the relationship between topology and multipartite entanglement.  While this work and others have established relationships between topological properties and entanglement structure, the exact relationship between the two remains unclear.  As recent work in high energy physics has revealed intimate links between entanglement and the shape spacetime, ideas such as topological data analysis may reveal new insights into the structure of emergent space time in quantum field theories.  As a result further study into this topic may yield more than just a deeper understanding of multipartite entanglement: it may shed new light onto the structure of reality.

\section*{Acknowledgements}
We would like to thank Gregory Hamilton and Felix Leditzky for useful conversations and feedback on the ideas presented here.  This work was primarily supported by acknowledges funding from the Quantum Software Centre funded by NSERC and NW's contribution on this work was supported by the U.S. Department of Energy,
Office of Science, National Quantum Information Science
Research Centers, Co-design Center for Quantum Advantage
(C2QA) under Contract Number DE-SC0012704.

\appendix
\section{Appendix}
\begin{lemma}
\label{lemma:binomial_equation}
     $\sum_{i=0}^{m} (-1)^i \binom{n}{i} = (-1)^m \binom{n-1}{m}$ for integers $n,m \geq 0$ and $m \leq n$
\end{lemma}
\begin{proof}
\textbf{Base Case:} $n = 0$, $m=0$: In such case we have
\begin{equation}
    \sum_{i=0}^{m} (-1)^i \binom{n}{i} = \binom{n}{0} = 0 = (-1)^0 \binom{n-1}{0}
\end{equation}
Lets fix an integer $n$ and $m \leq n$. Assume Induction Hypothesis, $\sum_{i=0}^{m} (-1)^i \binom{n}{i} = (-1)^m \binom{n-1}{m}$. We need to prove the result for $m+1$.
\begin{flalign}
    \sum_{i=0}^{m+1} (-1)^i \binom{n}{i} &= \sum_{i=0}^{m} (-1)^i \binom{n}{i} - (-1)^m \binom{n}{m+1} \\
    &= (-1)^m \binom{n-1}{m} - (-1)^m \left( \binom{n-1}{m+1} + \binom{n-1}{m}\right) \\
    &= (-1)^{m+1}  \binom{n-1}{m+1}
\end{flalign}
\end{proof}
\begin{lemma}
\label{lemma:AME_vanishing_correltion}
For any $ n $-qubit AME state $ \rho = \ket{\psi}\bra{\psi} $ defined over the qubit set $\mathcal{A}$, the total correlation over subsystems of size $ \leq \left\lfloor \frac{n}{2} \right\rfloor $ is zero.
\end{lemma}
\begin{proof}
Let $ J = \{ J_1, J_2, \dots, J_k \} $ be a subset of qubits such that $ J \subseteq \mathcal{A} $ and $ |J| = k \leq \left\lfloor \frac{n}{2} \right\rfloor $. The entropy of subset $ J $ is given by $ \log 2^{|J|} = |J| = k $. The entropy of the individual qubits is 1, as the AME states by definition maximizes every entropic measure. The value of the total correlation function is:
\begin{equation}
C(J) = S(J_1) + S(J_2) + \dots + S(J_k) - S(J) = k - k = 0
\end{equation}
\end{proof}
 \begin{prop}
For any $ n $-qubit AME state $\rho$ defined over the qubit set $\mathcal{A}$, the average total correlation over all possible subsets, i.e.,

\begin{equation}
\frac{1}{2^n}\sum_{\substack{J \subseteq \mathcal{A}}} C(J)_{\rho}    
\end{equation}
scales as $ \theta(\sqrt{n}) $.
\end{prop}
\begin{proof}
    Let's first consider the case for even $n = 2k$ for integer $k$. Using \ref{lemma:AME_vanishing_correltion}, we get that 
    \begin{align}
        \frac{1}{2^n}\sum_{\substack{J \subseteq \mathcal{A}}} C(J)_{\rho} &= \frac{1}{2^{2k}}\sum_{i = k + 1}^{2k-1} \sum_{\substack{J \subset A \\ |J| = i}} C(J) \\
        &= \frac{1}{2^{2k}}\sum_{i = k + 1}^{2k-1} \binom{2k}{i} (i - (2k - i)) = \frac{1}{2^{2k-1}}\sum_{i = k + 1}^{2k-1} \binom{2k}{i}(i - k) \\
        &= \frac{1}{2^{2k-1}}\sum_{i = k + 1}^{2k-1} \binom{2k}{i}i -  \frac{k}{2^{2k-1}}\sum_{i = k + 1}^{2k-1}\binom{2k}{i} \\
        &= \frac{2k}{2^{2k-1}}\sum_{i = k + 1}^{2k-1} \binom{2k - 1}{i - 1} - \frac{k}{2^{2k-1}}\sum_{i = k + 1}^{2k-1}\binom{2k}{i} \\
        &= \frac{2k}{2^{2k-1}}\sum_{j = k}^{2k-2} \binom{2k - 1}{j} - \frac{k}{2^{2k-1}}\sum_{i = k + 1}^{2k-1}\binom{2k}{i} \\
        &= \frac{2k}{2^{2k-1}} \left( 2^{2k-2} - 1 \right) - \frac{k}{2^{2k-1}} \left( 2^{2k-1} - 1 - \frac{1}{2} \binom{2k}{k} \right) \\
        &= k - \frac{k}{2^{2k-2}} - k + \frac{k}{2^{2k-1}} + \frac{k}{2^{2k}}  \binom{2k}{k} \\
        &=  \frac{k}{2^{2k}} \binom{2k}{k} - \frac{k}{2^{2k-1}} 
    \end{align}
    Using Stirling approximation for central binomial coefficients, we get 
    \begin{flalign}
       \frac{\sqrt{k}}{2} - \frac{2k}{2^{2k}} &\leq \frac{1}{2^n}\sum_{\substack{J \subseteq \mathcal{A}}} C(J)_{\rho} \leq \frac{\sqrt{k}}{\sqrt{\pi}}- \frac{2k}{2^{2k}} 
    \end{flalign}
    Asymptotically this approaches 
    \begin{flalign}
       \frac{\sqrt{k}}{2}  &\leq \frac{1}{2^n}\sum_{\substack{J \subseteq \mathcal{A}}} C(J)_{\rho} \leq \frac{\sqrt{k}}{\sqrt{\pi}} 
    \end{flalign}
    For the odd case, a similar calculation follows for $n = 2k+1$ to get 
        \begin{flalign}
        \frac{1}{2^n}\sum_{\substack{J \subseteq \mathcal{A}}} C(J)_{\rho} &= \frac{1}{2^{2k+1}}\sum_{i = k + 1}^{2k} \sum_{\substack{J \subset A \\ |J| = i}} C(J) \\
        &= \frac{1}{2^{2k+1}}\sum_{i = k + 1}^{2k} \binom{2k+1}{i} (i - (2k + 1 - i)) \\
        &= \frac{1}{2^{2k}}\sum_{i = k + 1}^{2k} \binom{2k+1}{i}(i - k) -   \frac{1}{2^{2k+1}}\sum_{i = k + 1}^{2k } \binom{2k+1}{i}\\
        &= \frac{1}{2^{2k}}\sum_{i = k + 1}^{2k} \binom{2k+1}{i}i - \frac{k}{2^{2k}}\sum_{i = k + 1}^{2k} \binom{2k+1}{i} -    \frac{1}{2^{2k+1}}\sum_{i = k + 1}^{2k } \binom{2k+1}{i} \\
        &= \frac{2k + 1}{2^{2k}} \sum_{i = k + 1}^{2k} \binom{2k}{i - 1} - \frac{k}{2^{2k}}\sum_{i = k + 1}^{2k} \binom{2k+1}{i} -    \frac{1}{2^{2k+1}}\sum_{i = k + 1}^{2k } \binom{2k+1}{i} \\
        &= \frac{2k + 1}{2^{2k}} \sum_{j = k }^{2k - 1} \binom{2k}{j} - \frac{k}{2^{2k}}\sum_{i = k + 1}^{2k} \binom{2k+1}{i} -    \frac{1}{2^{2k+1}}\sum_{i = k + 1}^{2k } \binom{2k+1}{i} \\
        &= \frac{2k+1}{2^{2k}} \left(2^{2k - 1} - 1 + \frac{1}{2} \binom{2k}{k}\right) - \frac{k}{2^{2k}} \left( \frac{2^{2k+1} - 2}{2} \right) - \frac{1}{2^{2k + 1}} \left( \frac{2^{2k+1} - 2}{2} \right)\\
        &= k + \frac{1}{2} - \frac{2k}{2^{2k}} - \frac{1}{2^{2k}} + \frac{2k+1}{2^{2k+1}} \binom{2k}{k} - k +\frac{k}{2^{2k}} - \frac{1}{2} +\frac{1}{2^{2k+1}} \\ 
        &=\frac{2k+1}{2^{2k+1}} \binom{2k}{k} - \frac{k}{2^{2k}} - \frac{1}{2^{2k+1}}
    \end{flalign}
    Approximating 
    \begin{flalign}
       \frac{2k+1}{4\sqrt{k}} - \frac{k}{2^{2k}} - \frac{1}{2^{2k+1}} &\leq \frac{1}{2^n}\sum_{\substack{J \subseteq \mathcal{A}}} C(J)_{\rho} \leq \frac{2k+1}{2\sqrt{\pi k}} - \frac{k}{2^{2k}} - \frac{1}{2^{2k+1}} 
    \end{flalign}
    Asymptotically
    \begin{flalign}
       \frac{\sqrt{k}}{2}  &\leq \frac{1}{2^n}\sum_{\substack{J \subseteq \mathcal{A}}} C(J)_{\rho} \leq \frac{\sqrt{k}}{\sqrt{\pi }} 
    \end{flalign}
\end{proof}
\begin{prop}
For any $ n $-qubit AME state $\rho$ defined over the qubit set $\mathcal{A}$, the average sum of entropies over possible subsets or the average distillable entanglement
\begin{equation}
    \frac{1}{2^n}\sum_{\substack{J \subseteq \mathcal{A}}}  S(J)_{\rho}
\end{equation}
scales as $\Theta(n)$.
\end{prop}
\begin{proof}
We earlier showed the value of average distillable entanglement for AME states (see Proposition \ref{prop:ADE_iff_AME}). After performing the approximation, the equation simplifies to the form $ c_1 n - c_2 \sqrt{n} $, where $ c_1 $ and $ c_2 $ are constants. As a result, the measure scales linearly with $ n $, since the $ n $-term dominates.
\end{proof}

\bibliographystyle{unsrt}
\bibliography{reference}

\begin{thebibliography}{10}

\bibitem{Horodecki_2009}
Ryszard Horodecki, Paweł Horodecki, Michał Horodecki, and Karol Horodecki.
\newblock Quantum entanglement.
\newblock {\em Reviews of Modern Physics}, 81(2):865–942, June 2009.

\bibitem{BB84}
Charles~H. Bennett and Gilles Brassard.
\newblock Quantum cryptography: Public key distribution and coin tossing.
\newblock {\em Theoretical Computer Science}, 560:7--11, 2014.
\newblock Theoretical Aspects of Quantum Cryptography – celebrating 30 years of BB84.

\bibitem{Metger2021selftestingofsingle}
Tony Metger and Thomas Vidick.
\newblock Self-testing of a single quantum device under computational assumptions.
\newblock {\em {Quantum}}, 5:544, September 2021.

\bibitem{Devetak_2006_ERC}
Todd Brun, Igor Devetak, and Min-Hsiu Hsieh.
\newblock Correcting quantum errors with entanglement.
\newblock {\em Science}, 314(5798):436--439, 2006.

\bibitem{horodecki2024multipartiteentanglement}
Pawel Horodecki, Łukasz Rudnicki, and Karol Życzkowski.
\newblock Multipartite entanglement, 2024.

\bibitem{computational_topology_book}
Herbert Edelsbrunner and John Harer.
\newblock {\em Computational Topology: An Introduction}.
\newblock 01 2010.

\bibitem{hamilton2023probing}
G.A. Hamilton and Felix Leditzky.
\newblock Probing multipartite entanglement through persistent homology.
\newblock {\em Communications in Mathematical Physics}, 405:125, 2024.

\bibitem{diPierro_2018}
Alessandra di~Pierro, Stefano Mancini, Laleh Memarzadeh, and Riccardo Mengoni.
\newblock Homological analysis of multi-qubit entanglement.
\newblock {\em Europhysics Letters}, 123(3):30006, sep 2018.

\bibitem{mengoni2019persistenthomologyanalysismultiqubit}
Riccardo Mengoni, Alessandra~DI Pierro, Laleh Memarzadeh, and Stefano Mancini.
\newblock Persistent homology analysis of multiqubit entanglement, 2019.

\bibitem{bart2023}
Bart Olsthoorn.
\newblock Persistent homology of quantum entanglement.
\newblock {\em Phys. Rev. B}, 107:115174, Mar 2023.

\bibitem{Kietaev_topology_2006}
Alexei Kitaev and John Preskill.
\newblock Topological entanglement entropy.
\newblock {\em Phys. Rev. Lett.}, 96:110404, Mar 2006.

\bibitem{Hosur_quantum_channels_2016}
Pavan Hosur, Xiao-Liang Qi, Daniel~A. Roberts, and Beni Yoshida.
\newblock Chaos in quantum channels.
\newblock {\em Journal of High Energy Physics}, 2016(2), February 2016.

\bibitem{Scott_k-uniform}
A.~J. Scott.
\newblock Multipartite entanglement, quantum-error-correcting codes, and entangling power of quantum evolutions.
\newblock {\em Phys. Rev. A}, 69:052330, May 2004.

\bibitem{vedral2002relativeentropy}
V.~Vedral.
\newblock The role of relative entropy in quantum information theory.
\newblock {\em Rev. Mod. Phys.}, 74:197--234, Mar 2002.

\bibitem{Hatcher:478079}
Allen Hatcher.
\newblock {\em {Algebraic topology}}.
\newblock Cambridge Univ. Press, Cambridge, 2002.

\bibitem{meyer2001global}
David~A Meyer and Nolan~R Wallach.
\newblock Global entanglement in multiparticle systems.
\newblock {\em arXiv preprint quant-ph/0108104}, 2001.

\bibitem{Cerf_k-MM}
Ludovic Arnaud and Nicolas~J. Cerf.
\newblock Exploring pure quantum states with maximally mixed reductions.
\newblock {\em Phys. Rev. A}, 87:012319, Jan 2013.

\bibitem{Karol_k-uniform_construction}
Dardo Goyeneche and Karol \ifmmode~\dot{Z}\else \.{Z}\fi{}yczkowski.
\newblock Genuinely multipartite entangled states and orthogonal arrays.
\newblock {\em Phys. Rev. A}, 90:022316, Aug 2014.

\bibitem{MoorNormal_form}
Frank Verstraete, Jeroen Dehaene, and Bart De~Moor.
\newblock Normal forms and entanglement measures for multipartite quantum states.
\newblock {\em Phys. Rev. A}, 68:012103, Jul 2003.

\bibitem{Kraus_LU_Equivalence}
B.~Kraus.
\newblock Local unitary equivalence of multipartite pure states.
\newblock {\em Phys. Rev. Lett.}, 104:020504, Jan 2010.

\bibitem{preskill_holographic_codes}
Fernando Pastawski, Beni Yoshida, John Preskill, and Daniel Harlow.
\newblock Holographic quantum error-correcting codes: toy models for the bulk/boundary correspondence.
\newblock {\em Journal of High Energy Physics}, 2015:1--55, 06 2015.

\bibitem{Helwig_AME}
Wolfram {Helwig}.
\newblock {\em {Multipartite Entanglement: Transformations, Quantum Secret Sharing, Quantum Error Correction}}.
\newblock PhD thesis, University of Toronto, Canada, January 2014.

\bibitem{Helwig:2013ckb}
Wolfram Helwig and Wei Cui.
\newblock {Absolutely Maximally Entangled States: Existence and Applications}.
\newblock 6 2013.

\bibitem{helwig2013quditgraph}
Wolfram Helwig.
\newblock {Absolutely Maximally Entangled Qudit Graph States}.
\newblock 6 2013.

\bibitem{Nielsen_Chuang_2010}
Michael~A. Nielsen and Isaac~L. Chuang.
\newblock {\em Quantum Computation and Quantum Information: 10th Anniversary Edition}.
\newblock Cambridge University Press, 2010.

\bibitem{bennet_distillation}
Charles~H. Bennett, David~P. DiVincenzo, John~A. Smolin, and William~K. Wootters.
\newblock Mixed-state entanglement and quantum error correction.
\newblock {\em Phys. Rev. A}, 54:3824--3851, Nov 1996.

\bibitem{nielsen2006}
Michael~A. Nielsen.
\newblock Cluster-state quantum computation.
\newblock {\em Reports on Mathematical Physics}, 57(1):147–161, February 2006.

\bibitem{hoyer2006graph}
Peter H{{\o}}yer, Mehdi Mhalla, and Simon Perdrix.
\newblock {Resources Required for Preparing Graph States}.
\newblock In {\em {Algorithms and Computation 17th International Symposium, ISAAC 2006, Kolkata, India, December 18-20, 2006. Proceedings}}, volume 4288 of {\em Lecture Notes in Computer Science}, pages 638 -- 649, Kolkata, India, December 2006.

\bibitem{claudet2024covering}
Nathan Claudet and Simon Perdrix.
\newblock Covering a graph with minimal local sets.
\newblock In {\em International Workshop on Graph-Theoretic Concepts in Computer Science}, pages 136--150. Springer, 2024.

\bibitem{hein2006entanglement}
Marc Hein, Wolfgang D{\"u}r, Jens Eisert, Robert Raussendorf, Maarten Van~den Nest, and H-J Briegel.
\newblock Entanglement in graph states and its applications.
\newblock In {\em Quantum computers, algorithms and chaos}, pages 115--218. IOS Press, 2006.

\bibitem{guhne2009entanglement}
Otfried G{\"u}hne and G{\'e}za T{\'o}th.
\newblock Entanglement detection.
\newblock {\em Physics Reports}, 474(1-6):1--75, 2009.

\bibitem{Walter_multipartite_entanglement}
Michael Walter, David Gross, and Jens Eisert.
\newblock {Multi-partite entanglement}.
\newblock 12 2016.

\bibitem{hage1996island}
Per Hage and Frank Harary.
\newblock {\em Island networks: Communication, kinship, and classification structures in Oceania}.
\newblock Number~11. Cambridge University Press, 1996.

\bibitem{Linden_2002_2-Particle_RDM}
N.~Linden, S.~Popescu, and W.~K. Wootters.
\newblock Almost every pure state of three qubits is completely determined by its two-particle reduced density matrices.
\newblock {\em Phys. Rev. Lett.}, 89:207901, Oct 2002.

\bibitem{Walck_2008_GHZ_RDM}
Scott~N. Walck and David~W. Lyons.
\newblock Only $n$-qubit greenberger-horne-zeilinger states are undetermined by their reduced density matrices.
\newblock {\em Phys. Rev. Lett.}, 100:050501, Feb 2008.

\bibitem{Di_2004_RDM}
Lajos Di\'osi.
\newblock Three-party pure quantum states are determined by two two-party reduced states.
\newblock {\em Phys. Rev. A}, 70:010302, Jul 2004.

\bibitem{berry2023topology}
Dominic~W. Berry, Yuan Su, Casper Gyurik, Robbie King, Joao Basso, Alexander Del~Toro Barba, Abhishek Rajput, Nathan Wiebe, Vedran Dunjko, and Ryan Babbush.
\newblock Analyzing prospects for quantum advantage in topological data analysis.
\newblock {\em PRX Quantum}, 5:010319, Feb 2024.

\bibitem{sudevan2022kuniformgraphs}
Sowrabh Sudevan and Sourin Das.
\newblock $n$-qubit states with maximum entanglement across all bipartitions: A graph state approach, 2022.

\end{thebibliography}
\end{document}